\documentclass[12pt]{article}
\usepackage{amsmath, amssymb, amsfonts, amsthm}
\usepackage{color}
\title{Scott correction for large atoms and molecules in a self-generated
magnetic field}
\author{L\'aszl\'o Erd\H os
 \thanks{Partially supported by SFB-TR12 of
the German Science Foundation. {\text lerdos@math.lmu.de} }
\\Institute of Mathematics, University of Munich \\
Theresienstr. 39, D-80333 Munich, Germany \\
S\o ren Fournais \thanks{Work partially supported by the Lundbeck
  Foundation, the Danish Natural Science Research Council and the European 
Research Council under the
 European Community's Seventh Framework Program (FP7/2007--2013)/ERC grant
 agreement  202859.
{\text fournais@imf.au.dk}} \\ Department of Mathematical Sciences, Aarhus University\\
 Ny Munkegade 118, DK-8000 Aarhus, Denmark
\\ and \\
Jan Philip Solovej \thanks{Work partially supported
   by the Danish Natural Science Research Council and by a Mercator
   Guest Professorship from the German Science Foundation. {\text
solovej@math.ku.dk}}
\\ Department of Mathematics, University of Copenhagen\\
Universitetsparken 5, DK-2100 Copenhagen,
Denmark}

\date{May 2, 2011}

\newtheorem{theorem}{Theorem}[section]
\newtheorem{proposition}[theorem]{Proposition}

\newtheorem{lemma}[theorem]{Lemma}

\newtheorem{remark}[theorem]{Remark}
\numberwithin{equation}{section}

\newcommand{\rd}{{\rm d}}
\newcommand{\be}{\begin{equation}}
\newcommand{\ee}{\end{equation}}
\newcommand{\bey}{\begin{eqnarray}}
\newcommand{\eey}{\end{eqnarray}}
\newcommand{\beys}{\begin{eqnarray*}}
\newcommand{\eeys}{\end{eqnarray*}}

\DeclareMathOperator{\supp}{supp}

\newcommand{\bz}{{\bf z}}

\newcommand{\bsigma}{\mbox{\boldmath $\sigma$}}

\newcommand{\cT}{{\cal T}}
\newcommand{\cQ}{{\cal Q}}

\renewcommand{\iint}{\int \!\! \int}

\newcommand{\bR}{{\mathbb R}}
\newcommand{\bC}{{\mathbb C}}
\newcommand{\bbR}{{\bf R}}

\newcommand{\bbZ}{{\bf Z}}

\newcommand{\br}{{\bf r}}

\newcommand{\bN}{{\mathbb N}}

\newcommand{\ov}{\overline}

\newcommand{\e}{\varepsilon}

\newcommand{\Tr}{{\rm Tr\;}}
\newcommand{\tr}{{\rm Tr\;}}

\newcommand{\wh}{\widehat}
\newcommand{\wt}{\widetilde}

\newcommand{\cE}{{\cal E}}
\newcommand{\cP}{{\cal P}}

\newcommand{\cM}{{\cal M}}

\newcommand{\cH}{{\cal H}}

\newcommand{\al}{\alpha}

\newcommand{\pt}{\partial}

\newcommand{\om}{\omega}

\newcommand{\non}{\nonumber}

\begin{document}
\maketitle

\begin{abstract}
We consider a large  neutral  molecule
with total nuclear charge $Z$ in non-relativistic quantum
mechanics with a self-generated classical electromagnetic field.
To ensure stability, we assume that $Z\al^2\le \kappa_0$ for
a sufficiently small $\kappa_0$, where $\al$ denotes the fine structure constant.
We show that, in the simultaneous limit $Z\to\infty$, $\al\to 0$ such that
$\kappa =Z\al^2$ is fixed, the ground state energy of the system is given by a
two term expansion $c_1Z^{7/3} + c_2(\kappa) Z^2 + o(Z^2)$. The leading term
is given by the non-magnetic Thomas-Fermi theory. Our result  shows that
the magnetic field affects only the second (so-called Scott) term in
the expansion.

\end{abstract}



\medskip\noindent
{\it Running title:} Scott correction for self-generated fields

\section{Introduction and the main result}

We  introduce the molecular many-body Hamiltonian
of $N$ dynamical electrons and $M$ static nuclei in three space dimensions.
The electron coordinates are $x_1, x_2, \ldots, x_N\in \bR^3$,
the nuclei are located at ${\bf R}=(R_1, R_2, \ldots , R_M)\in \bR^{3M}.$
Let ${\bf Z}= (Z_1, Z_2, \ldots , Z_M)$ denote the nuclear
charges, $Z_j>0$, with total nuclear charge $Z=\sum_{k=1}^M Z_k$.
We assume that the system is neutral, i.e. $N=Z$, in particular $Z$ is integer.
The electrons are subject to a self-generated magnetic field, $B=\nabla \times A$,
where $A\in H^1(\bR^3, \bR^3)$ is a magnetic vector potential. The magnetic
field energy is
$$
   \frac{1}{8\pi \al^2}\int_{\bR^3} B^2 = \frac{1}{8\pi \al^2} \int_{\bR^3} |\nabla\times A|^2,
$$
where $\al$ is the fine structure constant.

In the non-relativistic approximation,
the kinetic energy operator of the $j$-th particle 
is given by the magnetic Schr\"odinger or the Pauli operator,
\be
  T^{(j)}(A)= (-i\nabla_{x_j}+A(x_j))^2, \quad \mbox{or}
\quad T^{(j)}(A) = \big[ \bsigma\cdot (-i\nabla_{x_j}+A(x_j))\big]^2,
\label{Th1}
\ee
depending on whether the particles are considered spinless or have spin-$\frac{1}{2}$.
Here $\bsigma =(\sigma_1, \sigma_2, \sigma_3)$
is the vector of Pauli matrices.
 The Schr\"odinger operator acts on the space $L^2(\bR^3)$, the 
Pauli operator  acts on $L^2(\bR^3, \bC^2)$.
We will work with the Pauli operator, the treatment of the 
magnetic Schr\"odinger operator is simpler and we will only 
comment on the modifications needed.

The electrostatic potential of the electrons is the difference 
of the nuclear attraction
$$
     V( {\bf Z}, {\bf R}, x) = \sum_{k=1}^M \frac{Z_k}{|x-R_k|},
$$
and electron-electron repulsion
$$
   \sum_{1\le i < j \le N}\frac{1}{|x_i-x_j|}.
$$
The total energy of the electrons is given by the
Hamiltonian
\be
   H_N({\bf Z}, {\bf R}, A): = \sum_{j=1}^N \big[ T^{(j)}(A)
 -  V( {\bf Z}, {\bf R}, x_j)\big]
  +  \sum_{1\le i < j \le N}\frac{1}{|x_i-x_j|}. 
\label{HMN}
\ee
This operator acts on the space of antisymmetric functions $\bigwedge_1^N\cH$, where
$\cH = L^2(\bR^3)\otimes \bC^2$ is the single particle Hilbert space.

For a given vector potential $A$, the ground state energy of the electrons
is given by
\be \label{eq:parten}
   E({\bf Z}, {\bf R}, A) : = 
\inf{\mbox{Spec} \; H_N({\bf Z}, {\bf R}, A) }
\ee
with $N=Z= \sum_k Z_k$.
By gauge invariance, $E({\bf Z}, {\bf R}, A)$ depends only on the magnetic field
$B=\nabla\times A$.
Considering the magnetic field dynamical, we focus on 
the absolute ground state energy of the system, that includes the 
field energy,
\be
    E({\bf Z}, {\bf R}, A)+
 \frac{1}{8\pi \al^2}\int_{\bR^3} |\nabla \times A|^2,
\label{toten1}
\ee
and we will minimize over all vector potentials.

Since we are interested in gauge invariant quantities (like energy, ground state
density), we can always choose a divergence free gauge, $\nabla\cdot A=0$.
In this case, the field energy is given by
\be
 \frac{1}{8\pi \al^2} \int_{\bR^3} |\nabla\times A|^2 =  
\frac{1}{8\pi \al^2} \int_{\bR^3} | \nabla \otimes A|^2.
\label{fielden}
\ee
Since the magnetic energy  will always be finite, we can also  assume 
that $A\in L^6(\bR^3)$ (see Appendix of \cite{FLL} for the existence
of such a gauge), and we thus have 
\be
   \nabla\cdot A =0,  \qquad \Big(\int_{\bR^3}  A^6\Big)^{1/3} 
\le C \int_{\bR^3} |\nabla \otimes A|^2=
 \int_{\bR^3} |\nabla\times A|^2
\label{gauge}
\ee
by the Sobolev inequality.

We will call a vector potential $A$ {\it admissible} if
$A\in L^6(\bR^3)$, $\nabla \otimes A \in L^2(\bR^3),$ and $\nabla\cdot A=0$. 
For admissible vector potentials,
the total energy is
$$
  \cE({\bf Z}, {\bf R}, A, \al): =   E({\bf Z}, {\bf R}, A)+
\frac{1}{8\pi \al^2} \int_{\bR^3} |\nabla \otimes A|^2,
$$
and the absolute ground state energy of the system is given by
\be
  E_{\rm abs}({\bf Z}, {\bf R},\al):=
  \inf_{A} \Big\{  \cE({\bf Z}, {\bf R}, A, \al) 
  \Big\},
\ee
where the infimum is taken over all admissible vector potentials
  $A$.

Our units are $\hbar^2(2me^2)^{-1}$ for the length, $2me^4\hbar^{-2}$
for the energy and $2mec\hbar^{-1}$ for the magnetic vector
potential, where $m$ is the electron mass, $e$ is the electron charge and
$\hbar$ is the Planck constant. 
 In these units, the only physical parameter
that  appears in the total  Hamiltonian \eqref{toten1}
 is the dimensionless fine structure constant $\al=e^2(\hbar c)^{-1} \sim \frac{1}{137}$.
We will assume that $\max_k Z_k\al^2\leq\kappa_0$ with some sufficiently small
universal constant $\kappa_0\le 1$
and we will investigate the simultaneous limit $Z\to \infty$, $\al\to 0$.

\medskip

The main result of this paper  is:

\begin{theorem}\label{thm:main} Fix $M\in \bN$. Let $\bz = (z_1, z_2, \ldots, z_M)$ with
$z_1, z_2, \ldots, z_M>0$, $\sum_{k=1}^M z_k =1$, and
$\br = (r_1, r_2, \ldots, r_M)\in \bR^{3M}$ with
$\min_{k\neq \ell} |r_k-r_\ell|\ge r_{min}$ for some $r_{min}>0$ be given.
With a positive real parameter $Z>0$, define
${\bf Z} = (Z_1, Z_2, \ldots, Z_M)$, $Z_k:= Zz_k$,
and ${\bf R} = Z^{-1/3}\br $ to be the charges and
 the locations of the nuclei. Then there exists
a constant $E^{{\rm TF}}(\bz, \br)$ and a universal 
(independent of $\bz, \br$ and $M$), monotonically decreasing function
$S: (0, \kappa_0]\to \bR$ with some universal $\kappa_0>0$
and with $\lim_{\kappa\to 0+} S(\kappa) =\frac{1}{8}$ such that
as $Z= \sum_{k=1}^M Z_k\to \infty$, $\al\to 0 $ with
$\max_k 8\pi Z_k\al^2\le \kappa_0$, we have
\be\label{eq:main}
    E_{\rm abs}({\bf Z}, {\bf R}, \al) = Z^{7/3} E^{{\rm TF}}(\bz, \br)
  + 2Z^2\sum_{k=1}^M z_k^2 S(8\pi Z_k\al^2) + o(Z^2).
\ee
\end{theorem}

\begin{remark} The constant $E^{{\rm TF}}(\bz, \br)$ is given
by the Thomas-Fermi theory for nonrelativistic molecules
without magnetic field, see below. The factor 2 in the Scott term is
due to the spin degeneracy.
\end{remark}

\begin{remark} Our current proof does not provide an effective
error term in \eqref{eq:main} although we conjecture that
$o(Z^2)$ can be replaced with $O(Z^{2-\eta})$ with some $\eta>0$
where the constant in the
error term depends only on $r_{min}$ and $ M$.
\end{remark}

\begin{remark} {F}rom the monotonicity of $E_{\rm abs}({\bf Z}, {\bf R},\al)$
in $\al$, it is obvious that the function $S$ is monotonically decreasing.
It is known that there is a finite critical constant $\kappa_{cr}>0$
such that the system is unstable if $\max_k Z_k \al^2 > \kappa_{cr}$,
in particular the restriction $\max_k Z_k \al^2 \le \kappa_0$ in the theorem 
is necessary but our threshold $\kappa_0$ is not optimal.
 We conjecture that $S(\kappa)$ is a strictly
decreasing function, but unfortunately our proof does not provide this
statement. In fact, we even cannot exclude the possibility that
$S(\kappa)$ is constant (=1/8) for all $\kappa$ up to the critical
value $\kappa_{cr}$ beyond which it is minus infinity.
\end{remark}

\begin{remark}
An energy expansion similar to \eqref{eq:main} 
was derived in a relativistic model (without magnetic fields)
in \cite{SSS}. In the relativistic case one must also
consider simultaneous limits $Z\to\infty$, $\al\to0$. In this
case, however, it is the combination $Z_k\al$ which must remain bounded
in contrast to Theorem~\ref{thm:main}, where a bound on $Z_k\al^2$ is required.
In the atomic case, $M=1$, an alternative proof of the relativistic
energy asymptotics was given in \cite{FSW1}.
\end{remark}

The scaling can be understood from Thomas-Fermi theory
which we recall briefly \cite{LS, L}. Let $0\le \varrho(x) \in L^{5/3}(\bR^3)
\cap L^1(\bR^3)$ then the Thomas-Fermi energy functional is defined
as
$$
 \cE^{{\rm TF}}(\varrho): = \frac{3}{5} (3\pi^2)^{2/3}\int_{\bR^3} \varrho^{5/3} - 
  \int_{\bR^3} V(\bz, \br, x) \varrho(x) \rd x + D(\varrho)
$$
with
$$
  D(\varrho):= D(\varrho,\varrho):=
 \frac{1}{2}\int_{\bR^3}\int_{\bR^3} \frac{\varrho(x)\varrho(y)}{|x-y|} \rd x \rd y
$$
(the coefficient in front of the kinetic energy takes
into account the spin degeneracy).
It is well known that the variational problem
\be
  \inf\Big\{  \cE^{{\rm TF}}(\varrho)\; : \;  \varrho\in L^1(\bR^3)\cap
L^{5/3}(\bR^3), \; \; \int \varrho = \sum_{k=1}^M z_k \Big\}
\label{variation}
\ee
has a unique, strictly positive minimizer, called the {\it Thomas-Fermi density} and
denoted by $\varrho^{{\rm TF}}(x) = \varrho^{{\rm TF}}(\br,\bz, x)$. The value of
the minimum, $E^{{\rm TF}}(\bz, \br): =\cE^{{\rm TF}}(\varrho^{{\rm TF}})$, is called
the {\it Thomas-Fermi energy}. The function
\be
   V^{{\rm TF}}(\bz,\br, x):= V(\bz,\br,x) - \varrho^{{\rm TF}}* |x|^{-1}
\label{VTFdef}
\ee
is called the {\it Thomas-Fermi potential}; it is strictly positive and
it solves the Thomas-Fermi equation
\be
     V^{{\rm TF}}_{\bz,\br} = 
 (3\pi^2)^{2/3} \big[\varrho^{{\rm TF}}_{\bz,\br}\big]^{2/3}.
\label{TFE}
\ee
Sometimes we will use the notation $V^{{\rm TF}}_{\bz, \br}(x)$ instead of $V^{{\rm TF}}(\bz, \br, x)$
and likewise for $\varrho^{{\rm TF}}_{\bz,\br}$.
The key quantities in the Thomas-Fermi theory have the following
scaling behavior:
\begin{align}
    V^{{\rm TF}}(\bz,\br, x) & = h^{-4} V^{{\rm TF}}( h^3\bz, h^{-1}\br, h^{-1}x) \label{Vscale}\\
   \varrho^{{\rm TF}}(\bz,\br, x) & = h^{-6} \varrho^{{\rm TF}}( h^3\bz, h^{-1}\br, h^{-1}x) \non\\
  E^{{\rm TF}}(\bz, \br) &= h^{-7} E^{{\rm TF}}(h^3\bz, h^{-1} \br) \non
\end{align}
for any $h>0$.
We also note that
the Thomas-Fermi energy defined as the minimal value of \eqref{variation}
can also be given by the phase-space integral of the classical
symbol with the Thomas-Fermi potential, i.e.
\be
 E^{{\rm TF}}( \bz,\br): =  2 \frac{1}{(2\pi)^3} \int_{\bR^3\times\bR^3}
 [p^2 - V^{{\rm TF}}_{ \bz,\br}(q)]_-\rd p\rd q - D(\varrho^{{\rm TF}}_{ \bz,\br}),
\label{ETF}
\ee
where $[a]_-= - \min\{ a, 0\}$ denotes the negative part of a real number or a selfadjoint
operator. The factor 2 in \eqref{ETF} accounts for the spin degeneracy.

\medskip
Suppose we
replace the many-body electrostatic potential in  $H_N({\bf Z},  {\bf R}, A)$
by its mean field approximation
as
$$
   - \sum_{j=1}^N V({\bf Z}, {\bf R}, x_j ) +  \sum_{1\le i <j\le N}
  \frac{1}{|x_i-x_j|} \approx - V^{{\rm TF}}_{{\bf Z}, {\bf R}}(x_j) - 
D(\varrho^{{\rm TF}}_{{\bf Z}, {\bf R}}).
$$
Then the absolute ground state energy $E_{\rm abs}({\bf Z}, {\bf R},\al)$
will be approximated by
\begin{align}\label{eq:Ztoh}
\inf_A\Big\{ \Tr\big[ & T(A) -  V^{{\rm TF}}_{{\bf Z}, {\bf R}}\big]_- +
\frac{1}{8\pi \al^2} \int_{\bR^3} |\nabla \otimes A|^2 -D(\varrho^{{\rm TF}}_{{\bf Z}, {\bf R}}) \Big\} \\
& \approx Z^{4/3} \inf_A\Big\{ \Tr\big[  T_h(A) -  V^{{\rm TF}}_{{\bf z}, {\bf r}}\big]_- +
\frac{1}{8\pi Z\al^2} h^{-2} \int_{\bR^3} |\nabla \otimes A|^2  \Big\}
- Z^{7/3}D(\varrho^{{\rm TF}}_{{\bf z}, {\bf r}}) ,
\non
\end{align}
by using the scaling property \eqref{Vscale} with the choice $h=Z^{-1/3}$
and introducing the notation 
$$
T_h(A) =[\bsigma\cdot (-ih\nabla +A)]^2.
$$
The  infimum in \eqref{eq:Ztoh} is taken over all admissible vector potentials  $A$,
but in fact taking infimum over  slightly different sets yields the
same result, see a more detailed discussion in Appendix A of \cite{EFS1}.

Note that the trace in 
 \eqref{eq:Ztoh} is computed in the spinor space $L^2(\bR^3, \bC^2)$. 
Any operator $\cM$ acting on $L^2(\bR^3)$ will be naturally identified with $\cM\otimes I$
acting on the spinor space $L^2(\bR^3)\otimes \bC^2$ and we will not distinguish
between $\cM$ and $\cM\otimes I$ in the notation. This will not cause any confusion
but we need to keep in mind that $\tr_{L^2(\bR^3)\otimes \bC^2} \cM = 2 \tr_{L^2(\bR^3)} \cM$. 
Unless indicated otherwise, the traces in 
this paper are computed on $L^2(\bR^3)\otimes \bC^2$.

We thus reduced the problem to a semiclassical analysis of
the Pauli operator with a self-generated magnetic field.
The leading term in the asymptotic expansion in negative powers of $h=Z^{-1/3}$
for the one particle energy
in \eqref{eq:Ztoh} is given by the Weyl term 
\begin{align}
 \inf_A\Big\{ \Tr\big[  T_h(A) -  V^{{\rm TF}}_{{\bf z}, {\bf r}}\big]_- +
 & \frac{1}{\kappa h^2} \int_{\bR^3} |\nabla \otimes A|^2  \Big\}
 \non\\
& = h^{-3}   \frac{2}{(2\pi)^3} \int_{\bR^3\times\bR^3}
 [p^2 - V^{{\rm TF}}_{ \bz,\br}(q)]_-\rd p\rd q  + o(h^{-3})
\label{weylsc}
\end{align}
as long as $\kappa := 8\pi Z\al^2\le \kappa_0$ with some sufficiently
small fixed positive $\kappa_0$
(in fact, it is sufficient that $\max_k 8\pi  Z_k\al^2\le \kappa_0$).
Notice that the leading term  does not depend on $\kappa$.
Via \eqref{ETF}
this produces the leading term asymptotics  $Z^{7/3}E^{\rm TF}( {\bf z}, {\bf r})
 + o(Z^{7/3})$ for
the many body ground state energy.

\medskip

A fundamental result of Lieb and Simon \cite{LS}  (see also \cite{L}) rigorously justified
this heuristics without magnetic field and with an effective error term.
Thus they proved the leading term asymptotics 
in \eqref{eq:main} for the ground state energy of
large atoms and molecules  without magnetic field.
The next order term, known as the Scott correction, 
is of order $Z^2$.
For the non-magnetic case it is explicitly given by
\be
     2\cdot \frac{Z^2}{8}\sum_{k=1}^M z_k^2  
\label{Z2}
\ee
(the additional factor $2$ is due to the spin degeneracy)
and it was rigorously proved for atoms in \cite{H, SW1}
and for molecules in \cite{IS}
(see also \cite{SW2, SW3, SS}).
The next term in the expansion of order $Z^{5/3}$ was obtained
in \cite{FS}.

It was established in \cite{ES3} that the
 inclusion of a self-generated magnetic field 
does not change the leading term asymptotics in \eqref{eq:main}.
The  main theorem in this paper, Theorem \ref{thm:main} shows   that the effect of the
magnetic field appears in the second (Scott) term in
the asymptotic expansion.

The proof has two main steps. First we reduce the interacting
many-body problem to an effective one-body semiclassical
problem as it was described in \eqref{eq:Ztoh}.
 This step   will
be done rigorously in Section \ref{sec:reduc}. The second
step is an accurate second order semiclassical
asymptotics for the magnetic problem. More
precisely, we will  show the following more accurate
version of \eqref{weylsc}:
\begin{align}
\label{secondsc}
 \inf_A\Big\{ \Tr\big[  T_h(A) - &  V^{{\rm TF}}_{{\bf z}, {\bf r}}\big]_- +
 \frac{1}{\kappa h^{2}} \int_{\bR^3} |\nabla \otimes A|^2  \Big\}
 \\
& = h^{-3}   \frac{2}{(2\pi)^3} \int_{\bR^3\times\bR^3}
 [p^2 - V^{{\rm TF}}_{ \bz,\br}(q)]_-\rd p\rd q +
  2 h^{-2}\sum_{k=1}^M z_k^2 S(8\pi Z_k\al^2) + o(h^{-2}), \non
\end{align}
with   $\kappa = 8\pi Z\al^2$ and under the assumption 
that $\max_k 8\pi Z_k\al^2 \le \kappa_0$
with some $\kappa_0$ that is sufficiently small.
Together with \eqref{ETF}, \eqref{eq:Ztoh} 
and $h=Z^{-1/3}$, \eqref{secondsc} will yield \eqref{eq:main}.

\medskip

The precise second order semiclassical expansion for 
\be
    \inf_A\Big\{ \Tr\big[  T_h(A) - V\big]_- +
 \frac{1}{\kappa h^{2}} \int_{\bR^3} |\nabla \otimes A|^2  \Big\}
\label{ssec}
\ee
is of interest itself for a general potential $V$ and with $\kappa\le \kappa_0$.
Under general conditions on $V$ the leading term is given by the
classical Weyl term as in \eqref{weylsc}.
A local version of this statement was proven in Theorem 1.3 of \cite{ES3}
(this theorem contains only the lower bound, the upper bound is trivial).
The global version was given in Theorem 2.2 of \cite{EFS1}, where 
the condition $\kappa \le\kappa_0$ could even be relaxed to $\kappa = o(h^{-1})$.

The subleading term in the expansion for \eqref{ssec} depends on the 
singularity structure of the potential. For $V\in C_0^\infty$ we proved
in Theorem 1.1 of \cite{EFS2} that the Weyl asymptotics holds with an error
$O(h^{-2+\eta})$. The main ingredient was a local semiclassical asymptotics
that we recall in Theorem~\ref{thm:scMain}.

Using this result for smooth potentials and a 
multiscale analysis, 
we will show in Section~\ref{sec:multi}
that for potentials with Coulomb singularities, a non-zero second order
term arises  from the non-semiclassical effects of the innermost shells
of the Hydrogen-like atoms. The precise form of this term will not
be as explicit as in the non-magnetic case, \eqref{Z2}, but
it will be given by a universal function $S$ which we will
describe in Theorem~\ref{thm:scott} and prove in Section~\ref{sec:scott}.

\section{Semiclassical results up to second order}\label{sec:multi}

In this section we are interested in  noninteracting fermions, each is subject to the
one-body Hamiltonian $H= T_h(A)-V$.
The total ground state energy of the system is given by
\be
 \cE(A):=\tr \big[T_h(A) - V \big]_{-} +
 \frac{1}{\kappa h^2} \int_{\bR^3} | \nabla\times A|^2
= \tr \big[T_h(A) - V \big]_{-} +
 \frac{1}{\kappa h^2} \int_{\bR^3} | \nabla \otimes A|^2,
\label{toten2}
\ee
where $\kappa>0$ is a parameter, and where the last equality uses that
$\nabla \cdot A =0$ which can be assumed by gauge invariance.

\medskip

We will assume that the potential has a multiscale structure. Intuitively, this means that
there exist two scaling functions, $f, \ell: \bR^3\to \bR_+$  such
that for any $u\in \bR^3$, within the ball $B_u(\ell(u))$ centered
at $u$ with radius $\ell(u)$, the size of $V$ is of order $f^2(u)$
and $V$ varies on scale $\ell(u)$.  Moreover, we also require
that  the continuous family of balls  $B_u(\ell(u))$
supports a regular partition of unity.
The following lemma states this condition precisely.
This statement was  proved in Theorem 22 of \cite{SS}
with an explicit construction\footnote{Multiscaling was introduced in
  semiclassical problems in \cite{IS} (see also \cite{Sob, Sob1})}.

We will use the notation $B_x(r)$ for the ball of radius $r$ and 
center at $x$ and if $x=0$, we use $B(r) = B_0(r)$.

\begin{lemma}[{\cite[Theorem 22]{SS}}]\label{lem:PartUnityMultScale}
 Fix a cutoff function $\psi\in C_0^\infty(\bR^3)$
supported in the unit ball $B(1)$ satisfying $\int\psi^2=1$. 
Let $\ell: \bR^3\to (0,1]$ be a $C^1$ function with $\|\nabla\ell\|_\infty<1$.
Let $J(x,u)$ be the Jacobian of the map $u\mapsto (x-u)/\ell(u)$ and we define
$$
   \psi_u(x) = \psi\Big( \frac{x-u}{\ell(u)}\Big) \sqrt{J(x,u)} \ell(u)^{3/2}.
$$
Then, for all $x\in\bR^3$,
\be
   \int_{\bR^3} \psi_u(x)^2 \ell(u)^{-3} \rd u =1,
\label{partun}
\ee
and for all multi-indices $n\in \bN^3$ we have
\be\label{psider}
  \| \pt^n \psi_u\|_\infty \le C_n \ell(u)^{-|n|} , \qquad |n|=n_1+n_2+n_3,
\ee
where $C_n$ depends on the derivatives of $\psi$ but is independent of $u$.  \qed
\end{lemma}

We will require that the potential satisfies
\be
     |\pt^n V(u)| \le C_n f(u)^{2} \ell(u)^{-|n|}
\label{Vbound}
\ee
for all $n \in \bN^3$
uniformly in $u$ in some domain $\Omega\subset \bR^3$.
In applications, $\Omega$ will exclude an $h$-neighborhood of the core 
of the Coulomb potentials.
For brevity, we will often use $\ell_u= \ell(u)$ and $f_u = f(u)$.

Inserting the partition of unity \eqref{partun}, by IMS formula 
and reallocation of the localization error, we have
\begin{align} \label{IMS3}
  \cE(A) = & \tr \Big[\int_{\bR^3} \frac{\rd u}{\ell_u^3}\Big(
   \psi_u [ T_h(A) - V]\psi_u - h^2|\nabla\psi_u|^2 \Big)\Big]_{-} +
 \frac{1}{\kappa h^2} \int_{\bR^3} | \nabla \otimes A|^2 \cr
  \ge & \tr \Big[\int_{\bR^3} \frac{\rd u}{\ell_u^3} 
\psi_u \big[ T_h(A) - V - C h^2|\nabla\psi_u|^2\big]\psi_u  \Big]_{-} +
 \frac{1}{\kappa h^2} \int_{\bR^3} | \nabla \otimes A|^2 \cr
  \ge  & \int_{\bR^3} \frac{\rd u}{\ell_u^3} \cE(A,  V_u^+, \psi_u),
\end{align}
where
$$
    V_u^+: = V + Ch^2|\nabla\psi_u|^2
$$
and
$$  \cE(A, V_u^+, \psi_u) :=\tr \Big[\psi_u[ T_h(A) - V_u^+]\psi_u  \Big]_{-} +
 \frac{1}{\kappa h^2} \int_{\bR^3} \psi_u^2| \nabla \otimes A|^2 .
$$
In \eqref{IMS3}  we used
 $\tr [\int O_u \rd u]_-\ge \int \tr [O_u]_-\rd u$
for any continuous family of operators $O_u$.

We will assume that $V_u^+$ satisfies the same bound \eqref{Vbound} as $V$, i.e.
\be
   h \le Cf_u\ell_u.
\label{fll}
\ee
For the Coulomb-like
singularity, $V(x) = 1/|x|$, we will choose the $\ell(u)=c|u|$
with some $c<1$  and $f_u= \ell_u^{-1/2}$ so that \eqref{Vbound}
holds. With this choice \eqref{fll} holds for $|u|\ge C h^2$,
so the multiscale analysis will work apart from a very small
neighborhood of the nucleus.

\bigskip

Next we recall our  local semiclassical result from \cite{EFS2}
on a model problem
living in the ball $B(\ell)$ of radius $\ell>0$.

\begin{theorem}[{Semiclassical asymptotics \cite[Theorem 3.1]{EFS2}}]\label{thm:scMain} 
There exist universal constants $n_0 \in \bN$ and $\e >0$ such
that the following is satisfied.
Let $\kappa_0,f,\ell, h_0>0$ and let $\kappa \le \kappa_0 f^{-2}\ell^{-1}$.
Let $\psi\in C_0^\infty(\bR^3)$ with $\supp \psi \in B(\ell)$
and let $V\in C^\infty(\ov{B}(\ell))$ be a real valued potential satisfying
\be
    |\pt^n \psi|\le C_n \ell^{-|n|}, \qquad |\pt^n V|\le C_n f^2\ell^{-|n|}
\label{derMain}
\ee 
for every multiindex $n$ with $|n|\le n_0$.
Then
\begin{multline}
   \Bigg| \inf_A \Big( \tr [\psi H_h(A) \psi]_- + \frac{1}{\kappa  h^2}
\int_{B(2\ell)} |\nabla \otimes A|^2\Big)
  - 2(2\pi h)^{-3}\int_{\bR^3\times\bR^3} \psi(q)^2\big[ p^2 - V(q)\big]_- \rd q \rd p 
\Bigg| \\ \le  C  h^{-2+\e} f^{4-\e} \ell^{2-\e} 
\label{locscMain}
\end{multline}
for any $h\le h_0 f \ell$.
The constant $C$ depends only on $\kappa_0$, $h_0$ and on the constants $C_n$,  in \eqref{derMain}. The factor 2 in front of
the semiclassical term accounts for the spin and it is present only 
for the Pauli case. \qed
\end{theorem}

\medskip
{\it Remark.}
By variation of $\kappa$, we obtain  from Theorem~\ref{thm:scMain}
the following estimate
\begin{align}
  \label{eq:21}
  \frac{1}{\kappa  h^2}
\int_{B(2\ell)} |\nabla \otimes A|^2\leq
 C  h^{-2+\e} f^{4-\e} \ell^{2-\e},
\end{align}
for (near) minimizing vector potentials $A$.

\bigskip

 The following result from \cite{EFS2} can be viewed as a partial  converse to \eqref{eq:21}
as it estimates the semiclassical error in terms of the magnetic field.
 Note that
the assumption in \eqref{eq:17} below is much weaker than \eqref{eq:21}.

\begin{theorem}[{\cite[Theorem 3.2]{EFS2}}]\label{thm:UpperSemiclass}
Let the assumptions be as in  Theorem~\ref{thm:scMain}
 and assume that $A$ satisfies the bound
\begin{align}\label{eq:17}
\int_{B(2\ell)} |\nabla \otimes A|^2\leq
 C h^{-2}  f^{4} \ell^3.
\end{align}
Then, with $\e$ from Theorem~\ref{thm:scMain} we have
\begin{align}\label{eq:USNew}
C h^{-2}f^3\ell^{3/2}
 \Big\{ \int_{B(2\ell)}& |\nabla \otimes A|^2 \Big\}^{1/2}  + C h^{-1}f^3\ell \non\\ &\ge
\tr [\psi H_h(A) \psi]_{-} -  2(2\pi h)^{-3}\int_{\bR^3\times\bR^3} 
\psi(q)^2\big[ p^2 - V(q)\big]_- \rd q \rd p \nonumber \\
  & 
\ge  C  h^{-2+\e} f^{4-\e} \ell^{2-\e} - Ch^{-2}f^2\ell  \int_{B(2\ell)} |\nabla \otimes A|^2,
\end{align}
where the constants may depend on $h_0$ and $\kappa_0$ and on the constant in \eqref{eq:17}.
\end{theorem}

\bigskip

Very close to  the nuclei (at a distance $h^3\sim Z^{-1}$)
 the semiclassical approximation is no longer valid
due to the  local Coulomb
singularity and the energy is given by a specific function $S$ that
describes the contribution of the innermost electrons. 
Since the nuclei are separated on the much larger semiclassical scale $h$,
this effect  is additive for different nuclei. Thus $S$ 
can be defined via the Hydrogen atom. The following theorem
defines $S$ and gives some of its properties. It will be proven in 
Section~\ref{sec:scott}.

\begin{theorem}\label{thm:scott} 
Let  $\phi \in C_0^\infty(\bR^3)$ be a cutoff function
with $\supp \phi \subset B(1)$, $\phi \equiv 1 $ on $B(1/2)$
and such that $\wt\phi:=(1-\phi^2)^{1/2}\in  C^\infty(\bR^3)$.
Define $\phi_R(x) = \phi(x/R)$ for some $R>0$. 
There is a universal critical constant $\kappa_{cr}>0$
such that for any $\kappa<\kappa_{cr}$  and for any fixed $0<\beta \le (2\kappa)^{-1}$
 the following limit exists, it is finite and it depends only on $\kappa$
(and it is independent of the choice of $\phi$ and $\beta$ satisfying the
conditions listed above):
\begin{multline}
\label{def:Skappa}
   \lim_{R\to\infty} \Bigg[ \inf_A \Big\{ \tr \Big[ \phi_R \Big( T_{h=1}(A) - 
 \frac{1}{|x|} \Big) \phi_R\Big]_- + \frac{1}{\kappa} \int_{B(R/4)} |\nabla\otimes A|^2 
 + \beta \int_{B(2R)\setminus B(R/4)} |\nabla\otimes A|^2 
\Big\} \\
 - 2(2\pi)^{-3}\int_{\bR^3\times \bR^3} \phi_R^2(q)
 \Big[ p^2 - \frac{1}{|q|}\Big]_- \rd p \rd q \Bigg] =: 2 S(\kappa).
\end{multline}
The limit is denoted by $2S(\kappa)$ following the convention in
the literature to indicate explicitly the factor 2 that accounts
for  the spin degeneracy.
The function $S(\kappa)$ is defined on $[0, \kappa_{cr})$, it is decreasing
and $S(0)=1/8$,
corresponding to the coefficient of the non-magnetic Scott correction.

Furthermore, there exists a sequence $\{ A_R\} \subset H^1({\mathbb R}^3)$ with 
$\supp A_R \subset B(R/4)$ and such that
\begin{multline}
\label{def:Skappa2}
   \lim_{R\to\infty} \Bigg[   \tr \Big[ \phi_R \Big( T_{h=1}(A_R) - 
 \frac{1}{|x|} \Big) \phi_R\Big]_- + \frac{1}{\kappa} \int_{B(R/4)} |\nabla\otimes A_R|^2 
 \\
 - 2(2\pi)^{-3}\int_{\bR^3\times \bR^3} \phi_R^2(q)
 \Big[ p^2 - \frac{1}{|q|}\Big]_- \rd p \rd q \Bigg] = 2S(\kappa).
\end{multline}

\end{theorem}

This way of introducing a Scott correction when one cannot calculate 
its explicit value was introduced in \cite{Sob} and was used later in
\cite{SS, SSS}. In the case of the relativistic Scott correction
an alternative method of characterizing the Scott term
was given in \cite{FSW1, FSW2}.

Equivalently, one could define $S(\kappa)$ via another limiting procedure,
which may look more canonical:
\begin{lemma}\label{S=S}
For the function $S(\kappa)$ defined in \eqref{def:Skappa} we have
\begin{multline}
\label{def:Skappa1}
   \lim_{\mu\to0^+} \Bigg[ \inf_A \Big\{ \tr \Big[  T_{h=1}(A) - 
 \frac{1}{|x|} +\mu\Big]_- + \frac{1}{\kappa} \int_{\bR^3} |\nabla\otimes A|^2 
\Big\} \\
 - 2(2\pi)^{-3}\int_{\bR^3\times \bR^3}
 \Big[ p^2 - \frac{1}{|q|}+\mu\Big]_- \rd p \rd q \Bigg] = 2S(\kappa).
\end{multline}
\end{lemma}
In the main arguments of this paper
we will always use the definition \eqref{def:Skappa} and only in Section~\ref{sec:equiv},
with the help of the main semiclassical theorem, Theorem~\ref{thm:scMainscott},
 we will  prove Lemma~\ref{S=S}.

\bigskip

Now we formulate the precise second order semiclassical
asymptotics for a potential with Coulomb like singularities
and with certain scaling properties that are satisfied by the
Thomas-Fermi potential $V^{{\rm TF}}$. We first specify the
properties of  $V^{{\rm TF}}$ that are used in the proof.

For any given $\br =(r_1, r_2, \ldots, r_M)\in \bR^{3M}$ and $\bz = (z_1, z_2, \ldots, z_M)\in \bR^M_+$
set
$$
  r_{min}: =  \min_{k\ne \ell}|r_k-r_\ell|
$$
$$
  d(x):= \min \{ |x-r_k|\; : \; k=1,2, \ldots , M\}
$$
and
\be\label{fdef}
    f(x): = \min \{ d(x)^{-1/2}, d(x)^{-2}\}.
\ee
If $M=1$, then we set $r_{min}=\infty$.
We say that a potential $V$ is of {\it Thomas-Fermi type} if it satisfies
the following two properties:
\begin{itemize}
\item[(i)] There exists $\mu\ge0$ such that for any multiindex $\alpha$ with $|\alpha|\le n_0$
(where the universal constant $n_0$ is given in Theorem~\ref{thm:scMain}) we have
\be
   \Big| \partial_x^\al \big[ V(\bz,\br, x)+\mu\big]\Big|\le C_\alpha f(x)^2 d(x)^{-|\al|},
\label{Vder}
\ee
where the constants $C_\al$ depend only $\al$, $M$ and $\max_k z_k$; 
\item[(ii)] 
For $|x-r_k|\le r_{min}/2$, we have
\be
   -C\le  V(\bz,\br, x)- \frac{z_k}{|x-r_k|} \le Cr_{min}^{-1} + C
\label{West}
\ee
where $C$ depends on $\bz$.
\end{itemize}

Then we have 

\begin{theorem}[{\cite[Theorem 7]{SS}}] The Thomas-Fermi potential $V=V^{{\rm TF}}$
satisfies the conditions \eqref{Vder} and \eqref{West}. \qed
\end{theorem}

With these ingredients in Section~\ref{sec:scscott} we will prove
the following second order semiclassical asymptotics:
\begin{theorem}[Semiclassical asymptotics with Scott term]\label{thm:scMainscott} 
There exist universal constants $n_0 \in \bN$ and $\kappa_0>0$  
such
that the following is satisfied.
Let $V$ be a real valued potential satisfying \eqref{Vder} and \eqref{West}.
Then
\begin{multline}
  \lim_{h\to0} \Bigg| \inf_A \Big( \tr [T_h(A)- V]_- + \frac{1}{\kappa  h^2}
\int_{\bR^3} |\nabla \otimes A|^2\Big) 
\\
  - 2(2\pi h)^{-3}\int_{\bR^3\times\bR^3}\big[ p^2 - V(q)\big]_- \rd q \rd p 
- 2h^{-2}\sum_{k=1}^M z_k^2 S(z_k\kappa) \Bigg|  =0
\label{locscMain1}
\end{multline}
for any $0< \kappa\le\kappa_0$.

Moreover, there exist some  $\e>0$ and  $h_0>0$ and there exist an admissible vector potential $A$
and a density matrix $\gamma$, whose density $\varrho_\gamma$
satisfies
\be
\int\varrho_\gamma \le \frac{1}{3\pi^2} h^{-3}\int [V]_-^{3/2}
 + Ch^{-2+\e}
\label{densitycontroll}
\ee
and
\be
  D\Big(\varrho_\gamma - (3\pi^2)^{-1} h^{-3} [V]_-^{3/2}\Big)
   \le Ch^{-5+\e}
\label{densitycontroll65}
\ee
for any $0<h\le h_0$, 
such that
\begin{align}
\tr    [T_h(A) &- V]\gamma + \frac{1}{\kappa  h^2}
\int |\nabla \otimes A|^2 \label{trialenergy}\\
  &\le  2(2\pi h)^{-3}\iint \big[ p^2 - V(q)\big]_- \rd q \rd p
 +  2h^{-2}\sum_{k=1}^M z_k^2 S(z_k\kappa) + o(h^{-2}). \non
\end{align}
The constants $C$ in the right hand side of these
estimates depend only on $\kappa_0$, $h_0$ 
and on the constants  in \eqref{Vder} and \eqref{West}. The factor 2 in front of
the semiclassical term accounts for the spin and it is present only 
for the Pauli case.

\end{theorem}

{\it Convention:}  
All integrals, unless specified 
otherwise, are on $\bR^3$.

\bigskip

\section{Proof of the Main Theorem~\ref{thm:main}}\label{sec:reduc}

In this section we complete the proof of our main theorem.

\subsection{Lower bound}

The first step is to reduce the many body problem to 
a one body problem. We will use the following
Lemma whose proof relies on the Lieb-Oxford inequality \cite{LO}.

\begin{lemma}\cite[Lemma~4.3]{ES3}
There is a universal constant $C_0>0$ such that for any $\Psi \in \bigwedge_{1}^N C_0^{\infty}({\mathbb R}^3)
 \otimes {\mathbb C}^2$ with $\| \Psi \|_2 = 1$, for any non-negative function 
$\rho\,:\,{\mathbb R}^3 \rightarrow {\mathbb R}$ with $D(\rho,\rho) < \infty,$
 for any compactly supported and admissible $A$ and for any $\delta >0$ we have
\begin{align}
\Big\langle \Psi, \Big[ \delta \sum_{i=1}^N T^{(i)}(A) + \sum_{i<j} \frac{1}{|x_i-x_j|} \Big] 
\Psi \Big\rangle &+ C_0 \int_{{\mathbb R}^3} |\nabla \times A|^2 \nonumber \\
&\geq
-D(\rho,\rho) + \Big\langle \Psi, \sum_{i=1}^N (\rho * |x_i|^{-1}) \Psi \Big\rangle
 - C \delta^{-1} N. \non \qed
\end{align}
\end{lemma}
Thus, applying this lemma for $\varrho=\varrho_\Psi$, the density of $\Psi$, we get
\begin{align}
\langle \Psi, & H_N(\bbZ, \bbR, A) \Psi\rangle \non\\
& \ge 
\Big\langle \Psi, \sum_{j=1}^N \big[ (1-\delta)T^{(j)}(A)
 - V(\bbZ, \bbR, x_j) \big]\Psi\Big\rangle + D(\varrho_\Psi) - C\delta^{-1} N 
-C_0 \int_{{\mathbb R}^3} |\nabla \times A|^2\non\\
&\ge \Big\langle \Psi, \sum_{j=1}^N \big[ (1-\delta)T^{(j)}(A)
 - V^{{\rm TF}}(\bbZ, \bbR, x_j) \big]\Psi\Big\rangle 
- D(\varrho^{{\rm TF}}_{\bbZ, \bbR}) - C\delta^{-1} N-C_0 \int_{{\mathbb R}^3} |\nabla \times A|^2  \non\\
&\ge \Tr \big[ (1-\delta)T(A)
 - V^{{\rm TF}}_{\bbZ, \bbR} \big]_-
- D(\varrho^{{\rm TF}}_{\bbZ, \bbR}) - C\delta^{-1} N-C_0 \int_{{\mathbb R}^3} |\nabla \times A|^2  ,  \non
\end{align}
where in the second step we used the definition \eqref{VTFdef} and
that $ D(\varrho_\Psi-
 \varrho^{{\rm TF}}_{\bbZ, \bbR})\ge0$. We now add the field energy 
\eqref{fielden} and absorb the $-C_0\int |\nabla \times A|^2$
term at the expense of factor $(1-\delta)$ by using
$$
   \frac{1}{8\pi\al^2} - C_0 \ge \frac{1-\delta}{8\pi\al^2}
$$
as long as $\delta\gg Z^{-1}$ and $ Z\al^2$ is bounded.
Now we use the
scaling properties \eqref{Vscale} of the Thomas-Fermi theory
 with $h:=Z^{-1/3}$ and $\kappa:= 8\pi Z\al^2$ to get
\begin{align}
  E_{\rm abs}& (\bbZ, \bbR, \al)  \non\\ 
& \ge Z^{7/3}\Bigg[ h^3 (1-2\delta)\inf_A \Big( 
 \Tr \big[ T_h(A)
 - V^{{\rm TF}}_{\bz, \br} \big]_- + 
\frac{1}{\kappa h^2}\int |\nabla\otimes A|^2\Big) - 
 D(\varrho^{{\rm TF}}_{\bz, \br})\Bigg] -   C\delta^{-1} N
 \non\\
& \quad + \delta Z^{7/3}h^3\inf_A \Big( 
 \Tr \big[ T_h(A)
 - 2V^{{\rm TF}}_{\bz, \br} \big]_- + 
\frac{1}{\kappa h^2}\int |\nabla\otimes A|^2\Big)
 \non\\ 
&\ge  Z^{7/3}\Bigg[
 \frac{2}{(2\pi)^{3}}\iint\Big[p^2 -  V^{{\rm TF}}_{\bz, \br}(q)\Big]_-
 \rd q \rd p 
+2h\sum_{k=1}^M z_k^2 S(z_k\kappa) -  o(h)  - 
 D(\varrho^{{\rm TF}}_{\bz, \br})\Bigg] \non\\
&\quad - C\delta Z^{7/3}-  C\delta^{-1} N
 \non \\
&\ge  Z^{7/3}\Bigg[ E^{\rm TF}_{\bz,\br}
+ 2h\sum_{k=1}^M z_k^2 S(8\pi Z_k\al^2) -C\delta- o(h) \Bigg] -   C\delta^{-1} N
 \non \\
&\ge  E^{\rm TF}_{\bbZ,\bbR}+ 2Z^2\sum_{k=1}^M z_k^2 S(8\pi Z_k\al^2) 
 - o(Z^2) - CZ^{2-1/6}.
\end{align}
{F}rom the second to the third line we used 
\eqref{locscMain1} from Theorem~\ref{thm:scMainscott}
twice. 
In the main term the $1-2\delta$ prefactor can be easily
removed for a lower bound since the term $\inf_A\big(\ldots\big)$ is non-positive.
In the error term we used that $\inf_A\big(\ldots\big)$ is bounded by $O(h^{-3})$
since $\int  [V^{{\rm TF}}_{\bz, \br}]^{5/2} \le C$
with a constant depending only on $M$ (see \eqref{Vder}).
Then we used \eqref{ETF} and finally the scaling
relation \eqref{Vscale} for the energy.
In the last step we also inserted the optimal $\delta= Z^{-5/6}$
and used $N=Z$. This completes the proof of the lower bound
in Theorem~\ref{thm:main}.

\subsection{Upper bound}

By Lieb's variational principle \cite{L2} and neglecting the
exchange term, after a rescaling by $h=Z^{-1/3}$
 the energy of the particles \eqref{eq:parten} can 
be estimated by
\begin{align}
E(\bbZ, \bbR, A) & \le
   Z^{4/3} \Big( \Tr \big[ T_h(A) - V(\bz, \br, \cdot)\big]\gamma 
  + Z D( Z^{-1}\varrho_\gamma) \Big)\non\\
 & =  Z^{4/3} \Big( \Tr \big[ T_h(A) - V^{\rm TF}_{\bz, \br}\big]\gamma 
  + Z D\big( Z^{-1}\varrho_\gamma- \varrho^{\rm TF}_{\bz, \br}\big)
  - Z D\big( \varrho^{\rm TF}_{\bz, \br}\big) \Big)
\end{align}
for any density matrix $\gamma$ on $L^2(\bR^3, \bC^2)$ with
density $\varrho_\gamma(x) ={\mbox{Tr}}_{\bC^2}\gamma(x,x)$
 and with $\Tr \gamma = \int \varrho(x) \le N= Z$. Adding the field energy,
 we have
\begin{align}
E_{\rm abs}(\bbZ, \bbR, \al) \le & \; Z^{4/3}\Big( 
 \Tr \big[ T_h(A) -  V^{\rm TF}_{\bz, \br}\big]\gamma  + \frac{1}{8\pi Z\al^2} h^{-2}
  \int |\nabla \otimes A|^2  \non\\
 &  \qquad + Z D\big( Z^{-1}\varrho_\gamma- \varrho^{\rm TF}_{\bz, \br}\big)
  - Z D\big( \varrho^{\rm TF}_{\bz, \br}\big) \Big)
\end{align}
for any admissible vector potential $A$.

Let $\kappa= 8\pi Z\al^2$, set $V= V^{\rm TF}_{\bz, \br}$
 and
choose an admissible vector potential $A$ and
a density matrix $\wt \gamma$ according to \eqref{densitycontroll},
\eqref{densitycontroll65} and \eqref{trialenergy}. In particular
$$
  \tr \wt\gamma = \int \varrho_{\wt \gamma} \le Z (1+ CZ^{-1/3-\e/3})
$$
using that 
$$
  \frac{1}{3\pi^2} \int [V^{\rm TF}_{\bz, \br}]_-^{3/2} 
= \int \varrho_{\bz, \br}^{\rm TF}  =1
$$
by  \eqref{TFE} and the constraint $\int\varrho = \sum_k z_k =1$ in
\eqref{variation}.
We thus define $\gamma =  (1+ CZ^{-1/3-\e/3})^{-1}\wt \gamma$ so that
the constraint
$\tr \gamma\le Z$ is satisfied. Moreover, from \eqref{densitycontroll65}
we have
$$
    D\big( Z^{-1} \varrho_{\wt\gamma}  -  \varrho_{\bz, \br}^{\rm TF}\big)
  \le C Z^{-1/3-\e/3}.
$$
{F}rom the triangle inequality for $\sqrt{D}$ we obtain
$$
    D\big( Z^{-1} \varrho_{\gamma}  -  \varrho_{\bz, \br}^{\rm TF}\big)\le
   C(1+ CZ^{-1/3-\e/3})^{2}
 D\big( Z^{-1} \varrho_{\wt\gamma}  -  \varrho_{\bz, \br}^{\rm TF}\big)
 +  CZ^{-2/3-2\e/3} D\big(  \varrho_{\bz, \br}^{\rm TF}\big) \le  C Z^{-1/3-\e/3}
$$
using that $D\big(  \varrho_{\bz, \br}^{\rm TF}\big)\le C$ and $\e<1$.
In summary, we obtained
\begin{align}
E_{\rm abs}(\bbZ, \bbR, \al) \le & \; Z^{4/3}\Big( 
 \Tr \big[ T_h(A) -  V^{\rm TF}_{\bz, \br}\big]\gamma  + \frac{1}{8\pi Z\al^2} h^{-2}
  \int |\nabla \otimes A|^2 - Z D\big( \varrho^{\rm TF}_{\bz, \br}\big) +
 C Z^{2/3-\e/3}  \Big) \non
\end{align}
for the vector potential $A$ from Theorem~\ref{thm:scMainscott}.
Since  \eqref{trialenergy}
holds for $\wt\gamma$ and the energy of $\gamma$ and $\wt\gamma$ 
differ by a factor $(1+ CZ^{-1/3-\e/3})$, we obtain, after the usual rescaling,
\begin{align}
E_{\rm abs}& (\bbZ, \bbR, \al)   \non\\
 &\le Z^{7/3} \Bigg[ 2(2\pi )^{-3}\iint
 \big[ p^2 - V^{\rm TF}_{\bz,\br}(q)\big]_- \rd q \rd p - D\big( \varrho^{\rm TF}_{\bz, \br}\big)\Bigg]
 +  2Z^2\sum_{k=1}^M z_k^2 S(z_k\kappa) 
 + o(Z^{2}). \non
\end{align}
Using the identity \eqref{ETF}, we
thus obtain the upper bound in  \eqref{eq:main} 
which completes the proof of Theorem~\ref{thm:main}. \qed

\section{The Scott term}\label{sec:scott}
In this section we give the proof of Theorem~\ref{thm:scott}.

\begin{proof}[Proof of Theorem~\ref{thm:scott}]
First we prove the existence of the limit 
\eqref{def:Skappa}.
Define, for $R>0$ (and with $\kappa, \beta$ as in the statement of the theorem) 
\begin{align}
{\mathcal E}_{R,\kappa,\beta}(A): &= 
\tr \Big[ \phi_R \Big( T_{h=1}(A) - 
 \frac{1}{|x|} \Big) \phi_R\Big]_- + \frac{1}{\kappa} \int_{B(R/4)} |\nabla\otimes A|^2 
 \label{cEdef}  \\
&\qquad + \beta \int_{B(2R)\setminus B(R/4)} |\nabla\otimes A|^2 
 - 2(2\pi)^{-3}\int_{\bR^3\times \bR^3} \phi_R^2(q)
 \Big[ p^2 - \frac{1}{|q|}\Big]_- \rd p \rd q  \non
\end{align}
and
\begin{align}\label{cSdef}
S(R,\kappa,\beta):= \frac{1}{2}\inf_A {\mathcal E}_{R,\kappa,\beta}(A) .
\end{align}
We will prove later in \eqref{Slow} that this infimum is not minus infinity.

\medskip
\noindent{\bf Step 1: A-priori upper bound.}\\
Upon inserting $A=0$ we get $S(R,\kappa,\beta) \leq S(R,0,\infty)$. Since we
 know that $\lim_{R\rightarrow \infty} S(R,0,\infty) = S(0)$ exists 
\cite[Lemma~4.3 with $\al=0$]{SSS}, we
obtain that $S(R,\kappa,\beta)$ is bounded from above uniformly in $R$. I.e.,
there exists a constant $K_0$ (independent of $R, \kappa$ and $\beta$) such that
\begin{align}\label{eq:UpperNonmagnetic}
S(R,\kappa,\beta) \leq S(R, 0, \infty) \leq K_0.
\end{align}

\medskip

\noindent{\bf Step 2: Lower bound and semiclassics.}\\
Consider $r< R/8$. Let $\wt \phi_{r}$ satisfy $\phi_{r}^2 + \wt
\phi_{r}^2 = 1$ and define $W_{r} = |\nabla \phi_{r}|^2 + |\nabla \wt\phi_{r}|^2$ 
and $\phi_{r,R} = \phi_R \wt \phi_r$.
Split
\begin{align}\label{splitting}
\phi_R \big(T_{h=1}(A) -  \frac{1}{|x|} \big)\phi_R&\geq \phi_r  \big(T_{h=1}(A)
 -  \frac{1}{|x|} \big)\phi_r+\phi_{r,R} \big(T_{h=1}(A) -  \frac{1}{|x|} \big)\phi_{r,R}- W_{r}\non\\
&=\phi_r  \big(T_{h=1}(A) -  \frac{1}{|x|} - W_{r}\big)\phi_r+\phi_{r,R}\big(T_{h=1}(A)
 -  \frac{1}{|x|} - W_{r}\big)\phi_{r,R}.
\end{align}
The second term will be estimated by borrowing a small part of the field energy.
We will use that  the local semiclassical result Theorem~\ref{thm:scMain} holds with
 any positive constant in front of the field energy. However, this regime has to
be treated with multiscaling. The proof of the following lemma is postponed
to the end of this section.

\begin{lemma}\label{lm:multiscalesc}  For any $r_0>0$ and $\delta>0$ and
for any $r, R$ satisfying $r_0\le r\le R/8$ we have
\begin{align}
\inf_A \Bigg\{\tr\Big[\phi_{r,R}\big(T_{h=1}(A) -  \frac{1}{|x|}-W_r\big)\phi_{r,R}\Big]_{-}
 &+ \delta \int_{B(2R)\setminus B(r/4)} |\nabla \otimes A|^2\Bigg\}\label{eq:multi} \\
&\geq 2 (2\pi)^{-3}\int \phi_{r,R}^2(x)\Big[ p^2  -  \frac{1}{|x|}\Big]_{-}
 \rd x \rd p
- C_{\delta, r_0} r^{-\e/2 }, \non
\end{align}
where $\e>0$ is the exponent obtained from Theorem~\ref{thm:scMain}
and the constant in the error term depends only on $r_0$ and $\delta$.
\end{lemma}

Therefore, combining \eqref{splitting} and \eqref{eq:multi}, we have
\begin{align}
 \tr\Big[\phi_R \big(T_{h=1}(A) -  \frac{1}{|x|}\big)\phi_R\Big]_-
 &\geq \tr\Big[\phi_r \big(T_{h=1}(A) -  \frac{1}{|x|}-W_r\big)\phi_r\Big]_- \non\\ 
& \quad +2 (2\pi)^{-3}\int \phi_{r,R}^2(x)\Big[ p^2 -  \frac{1}{|x|}\Big]_{-}
 \rd x \rd p\nonumber \\
&\quad-\delta \int_{B(2R)\setminus B(r/4)}|\nabla\otimes A|^2-C_{\delta,r_0} r^{-\e/2}.
\end{align}
Hence, combining the semiclassical integrals, for any $r_0\le r \le R/8$ we get
\begin{align}
\label{eq:LowerCompare}
{\mathcal E}_{R,\kappa,\beta}(A) &\geq \tr\Big[\phi_r \big(T_{h=1}(A) - 
 \frac{1}{|x|}- W_r\big)\phi_r\Big]_{-} + \frac{1}{\kappa}\int_{B(R/4)}|\nabla\otimes A|^2
+\beta \int_{B(2R)\setminus B(R/4)} |\nabla\otimes A|^2 \nonumber \\
&\quad-\delta \int_{B(2R)\setminus B(r/4)}|\nabla\otimes A|^2 \nonumber \\
&\quad
-2 (2\pi)^{-3}\int \phi_{r}^2(x)\Big[ p^2 -  \frac{1}{|x|}\Big]_{-} 
\rd x \rd p-C_{\delta,r_0} r^{-\e/2}.
\end{align}
In the following, we will fix an $r_0>0$.

\medskip

\noindent {\bf Step 3: A-priori lower bound.}\\ 
We fix $r=r_0$.
On the ball of radius $r_0$ we use that if $\kappa\le\kappa^*$ 
where $\kappa^*$ is a small
universal constant,  then
\begin{align}
\label{eq:LowerInnerball}
\tr\Big[\phi_{r_0} \big(T_{h=1}(A) -  \frac{1}{|x|}-W_{r_0}\big)
 \phi_{r_0}\Big]_{-} + \frac{1}{2\kappa} \int_{B(2r_0)}|\nabla\otimes
A|^2 \geq - K_1
\end{align}
and 
\begin{align}\label{eq:LowerInnerball2}
-2 (2\pi)^{-3}\int \phi_{r_0}^2\Big[ p^2 -  \frac{1}{|x|}\Big]_{-}
 \rd x \rd p \geq -K_2,
\end{align}
where the constants only depend on $r_0$. The estimate \eqref{eq:LowerInnerball2}
 follows by simple integration. The other estimate \eqref{eq:LowerInnerball} 
is a consequence of (the proof of) \cite[Lemma~2.1] {ES3} in the case $Z=1$.

Inserting \eqref{eq:LowerInnerball} and \eqref{eq:LowerInnerball2}
 in \eqref{eq:LowerCompare}, we get
\begin{align}\label{eq:Inserted31/8}
{\mathcal E}_{R,\kappa,\beta}(A) \geq  
 (\frac{1}{2\kappa} - \delta) \int_{B(R/2)}|\nabla\otimes
A|^2 + (\beta-\delta) \int_{B(2R)\setminus B(R/2)}|\nabla\otimes
A|^2 -  K(r_0, \delta)
\end{align}
for any $\kappa\le \kappa^*$ and 
with a finite constant $K(r_0, \delta)$ depending only on $r_0$ and $\delta$.
Choosing $\delta=\beta \le (2\kappa)^{-1}$, this proves in particular that 
\be\label{Slow}
   S(R, \kappa, \beta)>-\infty,
\ee
and with the choice $\delta= (2\kappa^*)^{-1}$ we get
\be\label{Slow1}
  S(R, \kappa, (2\kappa)^{-1}) \ge -C(\kappa^*)
\ee
for some finite function $C(\kappa^*)<\infty$  depending only on $r_0$ and
$\kappa^*$ as long as $\kappa\le \kappa^*$.
\medskip

\noindent {\bf Step 4: Bound on the magnetic energy}.\\ 
{F}rom \eqref{eq:Inserted31/8}
we also obtain  that if $A$ is such that it yields a 
better energy on the large ball than no-magnetic field, i.e.
$$
   {\mathcal E}_{R,\kappa,\beta}(A) \le{\mathcal E}_{R,\kappa,\beta}(A=0),
$$ 
and $\kappa\le \kappa^*$,
then 
\be 
\int_{B(2R)} |\nabla\otimes A|^2\le C( r_0, \beta)
\label{fieldbound}
\ee
with some constant depending only on $r_0$ and $\beta$.
This follows from \eqref{eq:Inserted31/8} (taking $\delta =\delta_0(\beta):=
\frac{1}{2} \min((2\kappa^*)^{-1}, \beta)$)
 because then we have
$$ 
\frac{\delta_0(\beta)}{2} \int_{B(2R)} | \nabla \otimes A|^2 \leq {\mathcal E}_{R,\kappa,\beta}(A=0) +
 K(r_0, \delta_0(\beta)),
$$
and ${\mathcal E}_{R,\kappa,\beta}(A=0)$ is known to have a finite limit as $R \rightarrow \infty$,
independent of $\kappa, \beta$ 
\cite[Lemma~4.3]{SSS}.

\medskip

\noindent {\bf Step 5: Bound on the localization error}.\\ 
We will now remove the localisation error $W_r$ from \eqref{eq:LowerCompare}
 using the bound on the field energy. By the variational principle
and $W_r\le Cr^{-2}$,
 we can estimate
\begin{align}\label{eq:RemoveLoc}
\tr\Big[\phi_r \big(T_{h=1}(A) -  \frac{1}{|x|}- W_r\big)\phi_r\Big]_{-}
 &\geq \tr \Big[\phi_r \big(T_{h=1}(A) -  \frac{1}{|x|}\big)\phi_r\Big]_{-} 
\nonumber \\
&\quad - Cr^{-2} \tr {\bf 1}_{(-\infty,0)}\Big(
\phi_r \big(T_{h=1}(A) -  \frac{1}{|x|}- W_r\big)\phi_r\Big).
\end{align}
We have, with $g_r = {\bf 1}_{\{|x|\leq r\}}$,
\begin{align*}
\tr {\bf 1}_{(-\infty,0)}\Big(\phi_r \big(T_{h=1}(A) -  \frac{1}{|x|}- W_r\big)
\phi_r\Big) &\leq \tr {\bf 1}_{(-\infty,0)}\Big(\phi_r \Big[
(p+A)^2 -\frac{1}{|x|} - |B|- Cr^{-2}\Big]\phi_r\Big)  \\
&\leq \tr {\bf 1}_{(-\infty,0)}\Big((p+A)^2 -\Big[
\frac{1}{|x|} + |B|+ Cr^{-2}\Big]g_r\Big).
\end{align*}
Here we used the fact that we consider the strictly negative eigenvalues to get
 the last inequality. By the CLR estimate, we therefore have
\begin{align*}
\tr {\bf 1}_{(-\infty,0)}\Big(
\phi_r \big(T_{h=1}(A) -  \frac{1}{|x|}- W_r\big)\phi_r\Big) &\leq 
C \int_{\{|x|\leq r\}} \Big[ \frac{1}{|x|} + |B|+ Cr^{-2}\Big]^{3/2}\,\rd x \\
&\leq C' \int_{\{|x|\leq r\}} \Big[ \frac{1}{|x|^{3/2}} + |B|^{3/2}+ Cr^{-3}\Big]
\,\rd x \\
&\leq C'' \Big\{ r^{3/2} +  r^{3/4}\Big(\int_{\{|x|\leq r\}}|B|^2\Big)^{3/4} \Big\},
\end{align*}
where we used the H\"{o}lder inequality to get the last inequality. Using the uniform
 bound \eqref{fieldbound} on the field energy, we can therefore control the 
last term in \eqref{eq:RemoveLoc} for any $r\ge r_0$ as
\begin{align}
\label{eq:RemoveLoc2}
r^{-2} \tr {\bf 1}_{(-\infty,0)}\Big(\phi_r \big(T_{h=1}(A) - 
 \frac{1}{|x|}- W_r\big)\phi_r\Big) \leq C_{r_0} r^{-1/2}.
\end{align}
Thus
\begin{align}\label{eq:RemoveLocnew}
\tr\Big[\phi_r \big(T_{h=1}(A) -  \frac{1}{|x|}- W_r\big)\phi_r\Big]_{-}
 \geq \tr \Big[\phi_r \big(T_{h=1}(A) -  \frac{1}{|x|}\big)\phi_r\Big]_{-}  - C_{r_0} r^{-1/2}.
\end{align}

\medskip

\noindent {\bf Step 6: Monotonicity of the energy in the radius}.

Combining \eqref{eq:LowerCompare}, \eqref{eq:RemoveLocnew} and the definition 
\eqref{cEdef} with $R$ replaced with $r$ and $\beta$ replaced with $\beta'$ we get
\begin{align}
\cE_{R,\kappa,\beta}(A)& \ge \cE_{r,\kappa,\beta'}(A)  -  C_{\delta, r_0} r^{-\e/2} \non \\
&\;\; +\Bigg[\frac{1}{\kappa}\int_{B(R/4)\setminus B(r/4)}
 +\beta \int_{B(2R)\setminus B(R/4)}
-\beta' \int_{B(2r)\setminus B(r/4)}
-\delta \int_{B(2R)\setminus B(r/4)}\Bigg]|\nabla\otimes A|^2 \non \\
& \ge \cE_{r,\kappa,\beta'}(A)  + \frac{\beta}{2}\int_{B(2R)\setminus B(r/4)}|\nabla\otimes A|^2
 -  C_{\delta_0, r_0} r^{-\e/2} \label{mon1}
\end{align}
for any $\beta', \beta\le 1/(2\kappa)$ and
if we choose $\delta =\delta_0(\kappa,\beta):= \frac{1}{2}\min\{ (2\kappa)^{-1}, \beta\}$
and recall that $r\le R/8$. Notice that we even saved a part of the field energy
in $B(2R)\setminus B(r/4)$.
Taking infimum over all $A$, we have
\be
   S(R,\kappa,\beta) \ge  S(r,\kappa,\beta')  -  C_{\delta_0, r_0} r^{-\e/2}
\label{SSmon}
\ee
which shows that $R\to  S(R,\kappa,\beta)$ is essentially an increasing function
with a uniform upper bound \eqref{eq:UpperNonmagnetic}, hence it has a limit.
To be more precise,
define $S(\kappa,\beta): = \limsup_{R\to \infty} S(R,\kappa,\beta)$, then
for any $\eta>0$  there is a sufficiently large $r=r(\eta)>r_0$ such  that
$C_{\delta_0, r_0} r^{-\e/2}\le \eta/2$ and
$ S(r,\kappa,\beta) \ge S(\kappa,\beta) -\eta/2$.  Then  \eqref{SSmon}
implies that $S(R,\kappa,\beta) \ge  S(\kappa,\beta) -\eta$ for any $R\ge 8r(\eta)$.
Together with the definition of $ S(\kappa,\beta)$ this means that $S(\kappa,\beta)
=  \lim_{R\to \infty} S(R,\kappa,\beta)$.

\medskip
\noindent{\bf Step 7: Independence of the limit of $\beta$.}\\
Suppose  that $\beta < \beta' \le 1/(2\kappa)$. Then clearly $S(R,\kappa,\beta)\leq 
S(R,\kappa,\beta')$. Furthermore, by taking first the limit $R \rightarrow \infty$ 
in \eqref{SSmon}, then the limit $r \rightarrow \infty$ we obtain that
$$
\lim_{R\rightarrow \infty} S(R,\kappa, \beta)\geq \lim_{R\rightarrow \infty} S(R,\kappa, \beta').
$$
So we get that the limit $S(\kappa,\beta)$ is indeed independent of $\beta \le 1/(2\kappa)$.
Moreover, from \eqref{SSmon} it follows that
\be
    S(\kappa) \ge  S(r,\kappa,\beta)  -  C_{\delta_0, r_0} r^{-\e/2}
\label{lowerbd}
\ee
for any $r\ge r_0$, $\beta  \le 1/(2\kappa)$ and 
$\delta_0(\kappa,\beta):= \frac{1}{2}\min\{ (2\kappa)^{-1}, \beta\}$.
Combining this bound with \eqref{Slow1}, we obtain in particular that
\be
  \inf_{\kappa\le \kappa^*} S(\kappa) \ge C(\kappa^*, r_0)> -\infty
\label{infS}
\ee
for some constant depending only on $\kappa^*$ and $r_0$.
Together with the upper bound \eqref{eq:UpperNonmagnetic}
this shows  that $S(\kappa)$ is a bounded function for $\kappa\in (0,\kappa^*]$.

The fact that $S(0)=\frac{1}{8}$, i.e. the non-magnetic case,
has been proven before, see e.g. Lemma 4.3 in \cite{SSS} 
with the choice $\alpha=0$. (Note that in \cite{SSS} $S(0)=\frac{1}{4}$ is
stated but there the kinetic energy was $-\frac{1}{2}\Delta$, while 
our non-magnetic kinetic energy is $-\Delta$ which accounts for the apparent
discrepancy.)

It remains to prove that one can obtain $S(\kappa)$ by considering only vector potentials with small support.

\medskip
\noindent {\bf Step 8: Improved bound on field energy.}\\
Let $A_R$ be an (almost) minimizer of the variational problem
\eqref{cSdef}, i.e.
$$
  2S(R, \kappa,\beta) \ge \cE_{R,\kappa,\beta}(A_R) - R^{-\e/2}.
$$
Using \eqref{mon1} with the choice $r=R/8$, $\beta'=\beta$ and estimating
$\cE_{R/8, \kappa, \beta}(A_R)\ge 2S(R/8, \kappa,\beta)$, we  get
\be
   S(R, \kappa,\beta) \ge S(R/8, \kappa,\beta) +
 \frac{\beta}{2}\int_{B(2R)\setminus B(R/32)}|\nabla\otimes A_R|^2- C_{\delta_0, r_0} R^{-\e/2}
\label{SSbound}
\ee
for any $R\ge 8r_0$. Now letting $R\to\infty$, we conclude that
\begin{align}\label{eq:FEo}
  \lim_{R\to\infty}  
\frac{\beta}{2}\int_{B(2R)\setminus B(R/32)}|\nabla\otimes A_R|^2 = 0.
\end{align}

\medskip
\noindent {\bf Step 9: Upper bound on $S(R, \kappa, \beta)$}.\\ 
Fix $r = R/8$ and $\kappa, \beta$.
By the definition of $S(r, \kappa, \beta)$, there exists
a vector potential $A_r$ such that
\be
   2S(r, \kappa, \beta) \ge \cE_{r,\kappa, \beta}(A_r) - r^{-1},
\label{almin}
\ee
and we can assume that
\be
   \int_{B(2r)\setminus B(r/4)} A_r=0
\label{avg}
\ee
by adding a constant to $A_r$ if necessary.
Finally, by \eqref{eq:FEo} we have
\begin{align}\label{eq:er_to_zero}
e(r) := \int_{B(2r)\setminus B(r/32)} |\nabla \otimes A_r|^2 = o(1),
\end{align}
as $r \rightarrow \infty$.

Furthermore, 
there exists  a density matrix $\gamma_r$
such that
\begin{align} \label{eq1}
   2S(r, \kappa, \beta) \ge & \tr \phi_r\gamma_r\phi_r\Big( T_{h=1}(A_r) - 
 \frac{1}{|x|} \Big) + \frac{1}{\kappa} \int_{B(r/4)} |\nabla\otimes A_r|^2    \\
& + \beta \int_{B(2r)\setminus B(r/4)} |\nabla\otimes A_r|^2 
 - 2(2\pi)^{-3}\int_{\bR^3\times \bR^3} \phi_r^2(q)
 \Big[ p^2 - \frac{1}{|q|}\Big]_- \rd p \rd q - 2 r^{-1}.  \non 
\end{align}
We define $A_r': = \phi_{2r} A_r$, then $A_r'=A_r$ on the support of $\phi_r$.
Moreover,
\begin{align}
   \int_{B(2r)\setminus B(r/4)} |\nabla\otimes A_r'|^2 
  & \le Cr^{-2}  \int_{B(2r)\setminus B(r/4)} |A_r|^2 
  + \int_{B(2r)\setminus B(r/4)}\phi_{2r}^2 |\nabla\otimes A_r|^2 \non \\
 & \le C\int_{B(2r)\setminus B(r/4)}|\nabla\otimes A_r|^2  \le C e(r) \label{po}
\end{align}
with a universal constant $C$. Here we used the Poincar\'e inequality on the
ring $B(2r) \setminus B(r/4)$ which holds with a universal constant since the width of
the ring is comparable with its radius.

 Thus,  we have from \eqref{eq1}
\begin{align} \label{eq1bis}
   2S(r, \kappa, \beta) \ge & \tr \phi_r\gamma_r\phi_r\Big( T_{h=1}(A_r') - 
 \frac{1}{|x|} \Big) + \frac{1}{\kappa} \int_{B(r/4)} |\nabla\otimes A_r'|^2    \\
& 
 - 2(2\pi)^{-3}\int_{\bR^3\times \bR^3} \phi_r^2(q)
 \Big[ p^2 - \frac{1}{|q|}\Big]_- \rd p \rd q - C( r^{-1} + e(r)).  \non 
\end{align}

We use this $A_r'$ as a trial vector potential for $S(R, \kappa, \beta)$, i.e.
we have
\begin{align}
   2S(R, \kappa, \beta) \le & \tr \Big[ \phi_R \Big( T_{h=1}(A_r') - 
 \frac{1}{|x|} \Big) \phi_R\Big]_- + \frac{1}{\kappa} \int_{B(2r)} |\nabla\otimes A_r'|^2 
 \label{Sbeta}  \\
& - 2(2\pi)^{-3}\int_{\bR^3\times \bR^3} \phi_R^2(q)
 \Big[ p^2 - \frac{1}{|q|}\Big]_- \rd p \rd q . \non
\end{align}
The field energy integral (with coefficient $\kappa^{-1}$ in \eqref{cEdef}) can
be restricted to $B(2r)$ and the second
field energy integral is absent since $A'_r$ is supported on $B(2r)$.

Now we construct a suitable trial density matrix $\gamma$. It will have the form
$$
\gamma: = \phi_r\gamma_r\phi_r + \phi_{r,R} \wt \gamma \phi_{r,R}, 
$$
with $0 \leq \wt \gamma \leq 1$ to be chosen below.
Inserting $\gamma$ into \eqref{Sbeta} and using the support properties of
the cutoff functions, we obtain 
\begin{align}\label{Sbeta1}
   2S(R, \kappa, \beta) \le & {\mathcal E}_{R,\kappa,\beta}(A_r') \non\\
   \le & \tr \phi_r\gamma_r \phi_r \Big( T_{h=1}(A_r') - 
 \frac{1}{|x|} \Big)   + \tr \phi_{r,R} \wt \gamma \phi_{r,R} \Big( T_{h=1}(A_r') - 
 \frac{1}{|x|} \Big) 
   \non\\
& + \frac{1}{\kappa} \int_{B(2r)} |\nabla\otimes A_r'|^2 
- 2(2\pi)^{-3}\int_{\bR^3\times \bR^3} \phi_R^2(q)
 \Big[ p^2 - \frac{1}{|q|}\Big]_- \rd p \rd q  \non \\
 \le &  2S(r, \kappa, \beta)
 + C( e(r)+r^{-1}) + \Delta(\wt \gamma),
\end{align}
where
\begin{align} 
 \Delta(\wt \gamma) :=  \tr \phi_{r,R} \wt \gamma \phi_{r,R} \Big( T_{h=1}(A_r') - 
 \frac{1}{|x|} \Big) 
- 2(2\pi)^{-3}\int_{\bR^3\times \bR^3} \phi_{r,R}^2(q) \Big[ p^2 - \frac{1}{|q|}\Big]_- \rd p \rd q.
\end{align}
We used \eqref{almin}, \eqref{eq1} and \eqref{po} in the last step of \eqref{Sbeta1}.

We will use Theorem~\ref{thm:UpperSemiclass} and multiscaling to prove
 that we may choose $\wt \gamma$ such that
\begin{align}\label{DeltaBound}
\Delta(\wt \gamma) \leq C(\sqrt{e(r)} + r^{-1/2}).
\end{align}
Inserting \eqref{DeltaBound} in \eqref{Sbeta1} and recalling that $e(r) \rightarrow 0$ 
by \eqref{eq:er_to_zero}
 we obtain \eqref{def:Skappa2} (with the choice $A_R = A_r'$). Here we used the choice
 that $r =R/8$ and the 
previous result that $S(\kappa) = \lim_{R\rightarrow \infty} S(R,\kappa,\beta)$ exists 
and is independent of $\beta$.

It thus remains to prove \eqref{DeltaBound} for a suitable choice of $\wt \gamma$.
Let $\psi_u$ be given as in Lemma~\ref{lem:PartUnityMultScale}
with the choice $\ell(u) =\ell_u: = \frac{1}{100}\sqrt{r_0^2+u^2}$.
Define
$$
  \wt \gamma: = \int_{\cP} \psi_u \gamma_u \psi_u \frac{\rd u}{\ell_u^3},
$$
where
$$
 \cP := \{ x\; : \; r/3< |x| <R\}
$$
and 
$$
\gamma_u = {\bf 1}_{(-\infty,0]}\Big[\psi_u 
  \phi_{r,R} \Big( T_{h=1}(A_r') -   \frac{1}{|x|} \Big)
  \phi_{r,R} \psi_u \Big].
$$
Clearly $0\le \wt \gamma\le 1$
using \eqref{partun}.

Inserting this choice of $\wt \gamma$ into \eqref{Sbeta} and using \eqref{partun}, we obtain 
\begin{align}\label{Sbeta1old}
\Delta(\wt \gamma) =
 \int_{\cP} \frac{\rd u}{\ell_u^3} \Big\{ &\tr  \Big[ \phi_{r,R} \psi_u 
  \Big( T_{h=1}(A_r') -   \frac{1}{|x|} \Big) \psi_u  \phi_{r,R}  \Big]_{-}\non\\
  &
- 2(2\pi)^{-3}\int_{\bR^3\times \bR^3} \phi_{r,R}^2(q) \psi_u^2(q) 
\Big[ p^2 - \frac{1}{|q|}\Big]_- \rd p \rd q \Big\}.
\end{align}
Notice that we could restrict the $\rd u$ integration to $\cP$ 
since otherwise $\psi_u   \phi_{r,R}$
vanishes due to the support properties of these functions.

We will use Theorem~\ref{thm:UpperSemiclass} for each $u$ with
$$
h = 1, \qquad f_u^{-2} = \ell_u.
$$
We have 
$$
h^2 f_u^{-4} \ell_u^{-3} \int_{B(u,2\ell_u)} |\nabla \otimes A_r'|^2 \leq \ell_u^{-1} \int_{B(2R)\setminus B(r/4)} |\nabla \otimes A_r'|^2 \leq C e(r),
$$
by \eqref{po}. Here we used that $\supp A_r' \subseteq B(2r)$ to get the last inequality.
So  \eqref{eq:17} is satisfied and we get from \eqref{eq:USNew} and \eqref{Sbeta1old} that
\begin{align}\label{eq:424}
\Delta(\wt \gamma) \leq 
C \int_{\cP} \frac{\rd u}{\ell_u^3} \Big\{ e(r)^{1/2} +
 \ell_u^{-1/2} \Big\} \leq C \Big(e(r)^{1/2} + r^{-1/2} \Big)
\end{align}
for $r \geq r_0$. This finishes the proof of \eqref{DeltaBound}.
\end{proof}

{\it Proof of Lemma~\ref{lm:multiscalesc}.}
 We choose $\ell(u) =\ell_u: = \frac{1}{100}\sqrt{r_0^2+u^2}$ and 
$f_u= \ell^{-1/2}_u$ for the scaling functions and define
the ring
$$
 \cP := \{ x\; : \; r/3< |x| <R\}
$$
 which supports $\phi_{r,R}$. Inserting the partition of unity \eqref{partun}
and reallocating the 
 localization error we get
  \begin{align}
\tr\Big[\phi_{r,R} & \big(T_1(A) -  \frac{1}{|x|}  -W_r\big)\phi_{r,R}\Big]_{-}
 \non\\
 &= \tr\Big[ \int_{\cP} \frac{\rd u}{\ell_u^3} \Big( \psi_u  
 \phi_{r,R}\big(T_1(A) -  \frac{1}{|x|}-  W_r\big)\phi_{r,R} \psi_u  
 - |\nabla\psi_u|^2\phi_{r,R}^2\Big)
 \Big]_- \non \\
&\ge  \int_{\cP} \frac{\rd u}{\ell_u^3} \tr \Big[ \psi_u  
 \phi_{r,R}\Big(T_1(A) -  \frac{1}{|x|}-C\big(W_r +|\nabla\psi_u|^2\big)\Big)
\phi_{r,R} \psi_u \Big]_-. \non
\end{align}
Notice that we could restrict the $\rd u$ integration to $\cP$ 
since otherwise $\psi_u   \phi_{r,R}$
vanishes due to the support properties of these functions.  We also used that
 $\tr [\int O_u \rd \mu(u)]_-\ge \int \tr [O_u]_-\rd \mu(u)$
for any continuous family of operators $O_u$ and for any measure $\mu$.
We can also reallocate the field energy as
$$
  \int_{B(2R)\setminus B(r/4)} |\nabla \otimes A|^2 \ge  c\int_{\cP} \frac{\rd u}{\ell_u^3}
   \int_{B_u(2\ell_u)}|\nabla \otimes A|^2 
$$
with some positive universal constant $c$. Thus
  \begin{align}\label{ce}
\tr\Big[\phi_{r,R}\big(T_1(A) -  \frac{1}{|x|}  -W_r\big)\phi_{r,R}\Big]_{-}
 + \delta  \int_{B(2R)\setminus B(r/4)} |\nabla \otimes A|^2
 \ge  \int_{\cP} \frac{\rd u}{\ell_u^3} \cE_{r,R}(A, V^+, \psi_u), 
\end{align}
where we define
$$
 \cE_{r,R}(A, U, \psi_u): = 
\tr \Big[ \psi_u  
 \phi_{r,R}\big(T_1(A) - U\big)\phi_{r,R} \psi_u \Big]_- + c
\int_{B_u(2\ell_u)}|\nabla \otimes A|^2  \non
$$
for any potential $U$, 
and in the last step we used that 
$$
    \frac{1}{|x|}+C\big(W_r +  |\nabla\psi_u|^2\big) \le \frac{1}{|x|} \Big( 1 + \frac{C}{r}\Big)  
=:V^+(x).
$$
This inequality holds for any $u\in \cP$  by the support and scaling 
properties of $W_r$ and $\psi_u$
and by the estimate
$W_r\le Cr^{-2}{\bf 1}( r/2\le |x|\le r)\le Cr^{-1}/|x|$.

It is easy to see that $\psi=\psi_u\phi_{r,R}$ and $V=V^+$ satisfy  
the condition \eqref{derMain} in Theorem~\ref{thm:scMain} on the ball $B_u(2\ell_u)$
(the theorem was formulated for balls about the origin but it clearly holds
for balls with different center). We can choose $h_0\ge 1$ and  $\kappa_0$ 
 sufficiently large  so that $c\delta\ge \kappa_0^{-1}$ and 
$ 1\le h_0 f_u\ell_u = \frac{1}{10}h_0(1+u^2)^{1/4}$
are satisfied for all  $u\in \cP$ and $r\ge r_0$.
Thus, Theorem~\ref{thm:scMain} with $h=1$ gives 
$$
   \cE_{r,R}(A, V^+, \psi_u)\ge  2(2\pi)^{-3}\iint [(\psi_u\phi_{r,R})(x)]^2
[p^2- V^+(x)]_-\rd x\rd p
 - C_{\delta, r_0} \ell^{-\e/2}_u,
$$
where the constant $C_{\delta, r_0}$ is independent of $u$ but it depends on $\delta$ 
and on $r_0$ (via the
choice of $\kappa_0$ and $h_0$).

Inserting this bound into \eqref{ce}, we get
 \begin{align}
\tr\Big[\phi_{r,R}\big(T_1(A)& -  \frac{1}{|x|}  -W_r\big)\phi_{r,R}\Big]_{-}
 + \delta   \int_{B(2R)\setminus B(r/4)} |\nabla \otimes A|^2\non \\
& \ge  2(2\pi )^{-3}\iint \phi_{r,R}(x)^2
[p^2- V^+(x)]_-\rd x\rd p -  C_{\delta, r_0}  \int_{\cP}
 \frac{\rd u}{\ell_u^{3+\e/2}}, \non 
\end{align}
where in the first term we extended the $\rd u$ integration from $\cP$
to $\bR^3$  and we used \eqref{partun}. The last integral is bounded by $C_{\delta, r_0}r^{-\e/2}$,
 uniformly in $R$.

Finally, we can remove the localization errors from $V^+$, since
\begin{align}
\iint \phi_{r,R}(x)^2
\Big([p^2- V^+(x)]_- & -[p^2- |x|^{-1}]_-\Big) \rd x\rd p  \non\\
&\le C \int_{|x|\ge r/2}  \frac{1}{|x|^{5/2}} \Big[ \Big( 1+ \frac{C}{r}\Big)^{5/2}-1\Big]
\rd x \non \\
&\le C r^{-1/2}
\end{align}
with a constant independent of $R$. This error can be absorbed into the
other error as $\e\le 1$. Thus we get \eqref{eq:multi}
and have proved Lemma~\ref{lm:multiscalesc}. \qed

\section{Semiclassics with Scott term }\label{sec:scscott}

In this section we prove Theorem~\ref{thm:scMainscott}.
The guiding principle follows the similar proofs in \cite{SS, SSS}.
We first divide the space into three regions. The first region
consists of disjoint balls of radius $r\sim h$ about the
nuclei. Here we will use the Scott asymptotics as described in 
Theorem~\ref{thm:scott}.
The second region is far away from the nuclei, at a distance $R\gtrsim h^{-1}$.
The contribution of this regime will be estimated by a simple
Lieb-Thirring inequality.
Finally, in the third intermediate region, we will use
an argument similar to
Lemma~\ref{lm:multiscalesc} which relied
on the multiscale decomposition and Theorem~\ref{thm:scMain}
on each domain. The precise decomposition is the following.

Choose two localization functions $\theta_\pm\in C^1(\bR)$
with the properties that $0\le \theta_\pm\le 1$, 
$\theta_-^2+\theta_+^2\equiv 1$, moreover $\theta_-(t)=1$
for $t<1/2$ and $\theta_-(t)=0$ for $t\ge 1$.
Recall that $d(x)$ denotes the distance of $x$ to the nearest
nucleus and $r_{min}$ is the minimal distance among the nuclei.

For any $r< r_{min}/4$ and $R>r_{min}$ we set
$$
   \phi_\pm(x) : =\theta_\pm\Big(\frac{d(x)}{r}\Big), \qquad
   \Phi_\pm(x) : =\theta_\pm\Big(\frac{d(x)}{R}\Big).
$$
Note that
$$
  \phi_-(x) =\sum_{k=1}^M \theta_{r, k}(x), \qquad \mbox{with}
  \quad \theta_{r, k}(x) = \theta_-( |x-r_k|/r).
$$
Assuming $h$ is small enough, we will choose
\be\label{rRchoice}
   r:= h^{1-\xi}, \qquad  R : =
\begin{cases}
 Ch^{-1/2} & \mbox{if $\mu=0$}\\
C R_\mu & \mbox{if $\mu\ne0$}
\end{cases}
\ee
with some small $\xi>0$.
Here the $\mu$-dependent constant $R_\mu>0$ is chosen such that $-V(x)\ge 0$ for
$d(x)\ge R_\mu/2$.
Clearly $\Phi_-^2+\Phi_+^2=1$,  $\phi_-^2+\phi_+^2=1$ and 
$$
\phi_-^2 +  \Phi_+^2+ \Phi_-^2\phi_+^2 =1.
$$
This latter partition of unity corresponds to the three regions we
described above.
The localization errors will be affordable in all regimes. Near
the nuclei, the localization errors of order $h^2r^{-2}= O(h^{2\xi})$ is
more than
one order of magnitude smaller than the size of the potential
in this regime, which is $|x-r_k|^{-1}\sim h^{-1}$, and the localization
error relative to the potential becomes  even weaker further away from the nuclei.
 The localization
error far away is of order $h^2 R^{-2}= h^3$ will be negligible
both in $L^{5/2}$ and $L^4$ norms as necessary for the magnetic Lieb-Thirring
inequality that we recall for convenience:
\begin{theorem}\cite{LLS}\label{thm:lls}
There exist a universal constant $C$ such that
for the semiclassical Pauli operator  $T_h(A)-V$ 
with a potential $V\in L^{5/2}(\bR^3)\cap L^4(\bR^3)$ and magnetic field $B=\nabla\times A
\in L^2(\bR^3)$ we have
\be
    \tr\big[ T_h(A)-V\big]_- \le Ch^{-3}\int \big[ V\big]_+^{5/2} 
 +  C\Big( h^{-2} \int |B|^2\Big)^{3/4}\Big( \int  \big[ V\big]_+^{4} \Big)^{1/4}.
\label{genlt}
\ee
\end{theorem}

\subsection{Lower bound}

For any $\e\in (0,1/4)$, we define
\begin{align}
  \cT_k(A) : & =  \tr \Big[ \theta_{r,k} \big( T_h(A)- V - Ch^2r^{-2}\big) \theta_{r,k} \Big]_- +
   \Big(1-\frac{\e}{2}\Big)\frac{1}{\kappa  h^2}
\int_{B_{r_k}(r/4)} |\nabla \otimes A|^2 \non\\ 
&\qquad +
 \frac{1}{4\kappa h^2}
\int_{B_{r_k}(2r)\setminus B_{r_k}(r/4)} |\nabla \otimes A|^2.
\end{align}
Using the IMS localization  and the fact that the balls $B_{r_k}(2r)$ are disjoint since
$r< r_{min}/4$,  we have
\begin{align}
 \tr [T_h(A) & - V]_- +   \frac{1}{\kappa  h^2}
\int_{\bR^3} |\nabla \otimes A|^2  \non\\ 
&\ge \sum_{k=1}^M  \cT_k(A)\non\\
& + \tr \Big[ \Phi_+ \big( T_h(A)- V - Ch^2W_{R}\big) \Phi_+ \Big]_- 
+ \frac{\e}{2\kappa  h^2} \int_{\bR^3} |\nabla \otimes A|^2 \label{3term} \\
& + \tr \Big[ \Phi_- \phi_+\big( T_h(A)- V - Ch^2W_{r,R}\big)\phi_+ \Phi_- \Big]_- 
+ \frac{1}{8\kappa  h^2} \int_{\bR^3\setminus \bigcup_k B_{r_k}(r/4)} |\nabla \otimes A|^2 \non
\end{align}
with some positive universal constant  $C$
and with $W_R:=  |\nabla \Phi_-|^2 + |\nabla \Phi_+|^2$ and
$W_{r,R}: = |\nabla \phi_-|^2 + |\nabla \phi_+|^2 + |\nabla \Phi_-|^2 + |\nabla \Phi_+|^2$.

\medskip

\underline{\bf First line in \eqref{3term}}

\medskip

Fix $k$ and recall from \eqref{West} that 
$$
   - V(x) \ge - \frac{z_k}{|x-r_k|} - Cr_{min}^{-1} - C
$$
on the support of $\theta_{r,k}$. Thus we can write
$$
   \cT_k(A) \ge \cT_k^{(1)}(A) + \cT_k^{(2)}(A) + \cT_k^{(3)}(A),
$$
where
\begin{align}
 \cT_k^{(1)}(A):=  &  \tr \Big[ \theta_{r,k} \big( (1-2\e) T_h(A)- (1-2\e)\frac{z_k}{|x-r_k|}
  \big) \theta_{r,k} \Big]_- \non\\ 
& +
   \frac{1-2\e}{\kappa  h^2}
\int_{B_{r_k}(2r)} |\nabla \otimes A|^2 + \frac{1-2\e}{8\kappa h^2} \int_{B_{r_k}(2r)\setminus 
B_{r_k}(r/4)} |\nabla \otimes A|^2
 \non\\
\cT_k^{(2)}(A):= &  \tr \Big[ \theta_{r,k} \big( \e T_h(A)
 -Ch^2r^{-2} - Cr_{min}^{-1} - C \big) \theta_{r,k} \Big]_- +
   \frac{\e}{4\kappa  h^2}
\int_{B_{r_k}(2r)} |\nabla \otimes A|^2. \non\\
\cT_k^{(3)}(A):= &  \tr \Big[ \theta_{r,k} \big( \e T_h(A)- 2\e\frac{z_k}{|x-r_k|} \big) \theta_{r,k} \Big]_- +
   \frac{\e}{4\kappa  h^2}
\int_{B_{r_k}(2r)} |\nabla \otimes A|^2. \non
\end{align}
After pulling out the common $(1-2\e)$ factor, 
shifting $r_k$ to the origin and rescaling, we obtain from
\eqref{def:Skappa}, with $\beta = (8\kappa)^{-1}$, that
$$
 \inf_A \cT_k^{(1)}(A) \ge  (1-2\e) \Bigg[ 2(2\pi h)^{-3} \int_{\bR^3\times\bR^3} \theta_{r,k}^2 
\Big[ p^2- \frac{z_k}{|q-r_k|}\Big]_-  \rd q\rd p + 2h^{-2} z_k^2 S(z_k\kappa)\Bigg] + o(h^{-2}).
$$
Here we made use of the fact that after rescaling the variable $R$ in \eqref{def:Skappa} 
becomes $rh^{-1}= h^{-\xi}$, so the limit $h\to0$ is equivalent to
the limit $R\to\infty$ in 
\eqref{def:Skappa}. 

The term in the square bracket can be bounded by
\begin{align}
\Bigg| 2(2\pi h)^{-3} \int_{\bR^3\times\bR^3} & \theta_{r,k}^2 
\Big[ p^2- \frac{z_k}{|q-r_k|}\Big]_-  \rd q\rd p + 2h^{-2} z_k^2 S(z_k\kappa)\Bigg| \non\\
\le & Ch^{-3} \int_{|q-r_k|\le r} \frac{\rd q}{|q-r_k|^{5/2}} + Ch^{-2} \non\\
\le & Ch^{-3} r^{1/2} + Ch^{-2},
\end{align}
using $z_k\le 1$ and $S$ is bounded (if $\kappa\le \kappa_0$ is sufficiently small).
Thus
$$
 \inf_A \cT_k^{(1)}(A) \ge  \Bigg[ 2(2\pi h)^{-3} \int_{\bR^3\times\bR^3} \theta_{r,k}^2 
\Big[ p^2- \frac{z_k}{|q-r_k|}\Big]_-  \rd q\rd p + 2h^{-2} z_k^2 S(z_k\kappa)\Bigg]
 + o(h^{-2})
$$
as long as $\e r^{1/2}\ll h$.

The term $\cT_k^{(2)}(A)$ is estimated by the magnetic Lieb-Thirring inequality
\eqref{genlt} which we need
in the following form:

\begin{lemma}  
Let $\phi\in C_0^\infty(\bR^3)$ be a cutoff function with $\mbox{supp}\;\phi\subset B(1)$,
$\phi\equiv 1$ on $B(1/2)$. Define $\phi_\ell(x)=\phi(x/\ell)$ for some $\ell>0$
and  let  $\Omega: = \mbox{supp}\, \phi_\ell$.
Then for any vector potential $A$ we have
\be
  \tr \big[ \phi_\ell (T_h(A)-V)\phi_\ell\big]_- + \frac{\beta}{h^2} \int_{B(2\ell)} |\nabla \otimes A|^2
  \ge  -Ch^{-3}\int_\Omega V_+^{5/2} - C\beta^{-3}\int_\Omega V_+^{4} 
\label{lth}
\ee
with some universal constant $C$.
\end{lemma}
{\it Proof.} Let $\wt\phi_\ell (x):= \wt \phi (x/\ell)$, where $\wt\phi \in C_0^\infty(\bR^3)$ is supported in
$B(3/2)$ and $\wt\phi\equiv 1$ on $B(1)$. 
Set $\langle A\rangle : = |B(2\ell)|^{-1}\int_{B(2\ell)} A$ and $\wt A: = ( A- \langle A\rangle)\wt\phi_\ell$.
By a gauge transformation and by the fact that $\wt A =  A- \langle A\rangle$ on the support of $\phi_\ell$,
we have
$$ 
\tr \big[ \phi_\ell (T_h(A)-V)\phi_\ell\big]_- =\tr \big[ \phi_\ell (T_h(\wt A)-V)\phi_\ell\big]_-.
$$ 
By the magnetic Lieb-Thirring inequality \eqref{genlt} we have
\be
   \tr \big[ \phi_\ell (T_h(\wt A)-V)\phi_\ell\big]_-  \ge -Ch^{-3}\int_\Omega V_+^{5/2}
  - C \Big( h^{-2}\int_{\bR^3} |\nabla\times \wt A|^2\Big)^{3/4}\Big( \int_\Omega V_+^4\Big)^{1/4}.
\label{mlt}
\ee
By $\mbox{supp} \, \wt A\subset B(2\ell)$ and by the Poincar\'e
inequality
$$
  \int_{\bR^3} |\nabla\times \wt A|^2 =  \int_{\bR^3} |\nabla\otimes \wt A|^2
\le C\int_{B(2\ell)} |\nabla \otimes A|^2 + C\ell^{-2}
   \int_{B(2\ell)} |A- \langle A\rangle|^2 \le  C\int_{B(2\ell)} |\nabla \otimes A|^2.
$$
Combining this with \eqref{mlt}  we obtain \eqref{lth}. $\Box$

\medskip

We return to the estimate of $\cT_k^{(2)}(A)$, and shifting $r_k$ to the origin and
using \eqref{lth} we obtain
\begin{align}
\cT_k^{(2)}(A)\ge & -C\e h^{-3} \int_{|x|\le r} \Big(  \e^{-1}\big[
h^{2} r^{-2} + r_{min}^{-1} + 1\big]\Big)^{5/2}
 - C\e  \int_{|x|\le r} \Big( \e^{-1}\big[
h^{2} r^{-2} + r_{min}^{-1} + 1\big]\Big)^4 \non\\
\ge & \; - C\e^{-3/2} h^{-3}\Big( h^5 r^{-2} +r^3\Big)
  -C\e^{-3} \Big( h^8 r^{-5} +r^3\Big) \ge -C\e^{-3/2}h^{-3\xi} - C\e^{-3} h^{3-3\xi} \non
\end{align}
with a constant $C$ depending on $\kappa_0$ and  $r_{min}$.
This error term is of order $o(h^2)$ as long as $\e\ge h$ and $\xi< 1/6$,
which we will assume from now on.

Finally, for the term $\cT^{(3)}_k$, after shifting, rescaling and using $z_k\le 1$
$$
 \inf_A \cT_k^{(3)}(A) \ge \e \inf_A \Bigg\{ \Tr \Big[ \theta_{d} 
\big( T_{h=1}(A) - \frac{2h^{-1} }{|x|} \big)\theta_{d}\Big]
 + \frac{1}{4\kappa h} \int_{B(2d)} |\nabla\otimes A|^2\Bigg\}
$$
with $d:=rh^{-1} = h^{-\xi}$ and $\theta_d(x)=\theta_-(x/d)$.
Now we will use the ``running energy scale'' argument as in the proof of Lemma 2.1
of \cite{ES3}, where $2h^{-1}$ plays the role of $Z$ in Lemma 2.1 of \cite{ES3},
the localization errors in (2.14) of \cite{ES3} are not present, and 
the key condition in \cite{ES3}  that $Z\al^2$ is sufficiently small
translates into $\kappa$ being sufficiently small. 
Using the final result (3.23) from \cite{ES3}, and noting
that the $d^{-2}$ term there was due to the localization error
that is not present now, we obtain
\be\label{t3}
 \inf_A \cT_k^{(3)}(A) \ge -C \e  h^{-5/2} d^{1/2},
\ee
which if $o(h^{-2})$ provided $\e\ll h^{(\xi+1)/2}$. These constraints, together with 
the previous conditions  $\e r^{1/2}\ll h$, $\e\ge h$ and $\xi<1/6$
leave plenty of room, choosing for example $\xi=\frac{1}{10}$ and $\e = h^{3/4}$.
In summary, for the first line  in \eqref{3term} we proved that
\begin{align}
\inf_A \sum_{k=1}^M \cT_k(A) \ge & \sum_{k=1}^M
 \Bigg[ 2(2\pi h)^{-3} \int_{\bR^3\times\bR^3} \theta_{r,k}^2 
\Big[ p^2- \frac{z_k}{|q-r_k|}\Big]_-  \rd q\rd p + 2h^{-2} z_k^2 S(z_k\kappa)\Bigg]
+ o(h^{-2}).
\end{align}

\bigskip

\underline{\bf Second line in \eqref{3term}}

\bigskip
By the magnetic Lieb-Thirring inequality from \eqref{lth} 
 we have
\begin{align}
\tr \Big[ \Phi_+ \big( T_h(A)- & V - Ch^2W_{R}\big) \Phi_+ \Big]_-
+ \frac{\e}{2\kappa  h^2} \int |\nabla \otimes A|^2  \non\\
\ge & - Ch^{-3} \int_{\Omega_+} [V+Ch^2W_{R}]_+^{5/2}
 - C\kappa^3 \e^{-3}\int_{\Omega_+} [V+Ch^2W_{R}]_+^{4},
\end{align}
where $\Omega_+: = \mbox{supp} \Phi_+$.
The contribution  of the $W_R$ terms is negligible using $\| W_R\|_\infty\le CR^{-2}$:
and that its support has a volume $CR^3$:
$$
    h^{-3} \int [h^2 W_R]^{5/2} + \kappa^3\e^{-3}\int [h^2 W_R]^{4}  
\le C h^2  R^{-2} +C\kappa_0^3 \e^{-3}
h^8 R^{-5}\le \begin{cases}  Ch^3 & \mbox{if $\mu=0$}\\
Ch^2  & \mbox{if $\mu\ne 0$}
\end{cases}
$$
with a constant that may depend on $\mu$.
The positive part of the potential  $[V(x)]_+$
is zero for $\mu\ne0$ in $\Omega_+$ and it  
 can be estimated by $f(x)^2 \le d(x)^{-4}$
 according to \eqref{Vder} if $\mu=0$, so
$$
   \int_{\Omega_+} \Big(h^{-3} [V+\mu]_+^{5/2} + 
\kappa^3\e^{-3} [V+\mu]_+^4  \Big) \le \begin{cases}
 C (h^{-3}  R^{-7} + \e^{-3}R^{-13} ) & \mbox{if $\mu=0$}\\
0  & \mbox{if $\mu\ne 0$}
\end{cases}
$$
with a constant depending on $\kappa_0$ and $M$. With the choice of \eqref{rRchoice}
and recalling $\e\ge h$,
the lower bound on the second line  in \eqref{3term} thus vanishes as $h\to 0$.

\medskip

\underline{\bf Third line in \eqref{3term}}
\medskip

This estimate will be very similar to the proof of Lemma~\ref{lm:multiscalesc},
so we will skip some details here.
 We choose $\ell(u) =\ell_u: = \frac{1}{100}\sqrt{r^2+d(u)^2}$ and 
$f_u= \min\{\ell^{-1/2}_u, \ell_u^{-2}\}$ for the scaling functions and
define the regime
\be
\label{def:Q}
   \cQ: = \{ x\; : \; |x|\le 2R, \; |x-r_k|\ge r/3,\; k=1,2, \ldots, M\}
\ee
which supports $\Phi_-\phi_+$. Inserting the partition of unity \eqref{partun}
and reallocating the localization error,  we have
\begin{align}
 \tr \Big[ \Phi_- \phi_+\big( T_h(A) & - V - Ch^2W_{r,R}\big)\phi_+ \Phi_- \Big]_- 
+ \frac{1}{8\kappa  h^2} \int_{\bR^3\setminus \bigcup_k B_{r_k}(r/4)} |\nabla \otimes A|^2\non\\
& \ge \int_\cQ \frac{\rd u}{\ell_u^3} \cE(A, V^+_u, \psi_u),\label{ed}
\end{align}
where we defined
$$
  \cE (A, U, \psi_u):=  \tr \Big[ \psi_u\Phi_- \phi_+\Big( T_h(A) 
 - U\Big)\phi_+ \Phi_-\psi_u
 \Big]_-  + \frac{c}{\kappa h^2} \int_{B_u(2\ell_u)} |\nabla\otimes A|^2
$$
for any potential $U$ and with some sufficiently small positive universal constant $c>0$,
 and we defined
$$
V_u^+(x):=V(x)+ Ch^2\big(W_{r,R}(x)+|\nabla\psi_u(x)|^2\big)
$$
on the support of $\psi_u$. We recall
from \eqref{Vder} that $V(x)\le Cf(x)^2\le Cf_u^2$
 on the support of $\psi_u$,
since $f(x)$, defined in \eqref{fdef},
is comparable with $f_u$. Thus
$$
|V_u^+(x)|\le Cf_u^2 +  Ch^2\ell_u^{-2}\le Cf_u^2\le Cf(x)^2, 
\qquad \mbox{on supp}\; \psi_u,
$$
where we distinguished the case $\mu=0$ and $\mu\ne0$.
In the latter case $|u|$ is bounded (depending on $\mu$)
and thus $f_u$ is bounded from below.
We also used that $h\le C\ell_u f_u= C\min\{ \ell^{1/2}_u, \ell_u^{-1}\}$,
which holds since $\ell_u$ is between a constant multiple of $r$ and $R$
if $u\in \cQ$.  
 Similar estimate holds for the derivatives of $V_u^+$, i.e.
the main condition \eqref{derMain} of Theorem~\ref{thm:scMain} is satisfied
for $V_u^+$ with scaling functions $\ell=\ell_u$ and $f=f_u$. 
The other condition, $\kappa\le \kappa_0f^{-2}_u\ell_u^{-1}$, is trivially satisfied.
Applying  Theorem~\ref{thm:scMain}, we get from  \eqref{ed} that
\begin{align}
  \inf_A\Bigg[ & \tr \Big[ \Phi_- \phi_+\big( T_h(A)  - V - Ch^2W_{r,R}\big)\phi_+ \Phi_- \Big]_- 
+ \frac{1}{8\kappa  h^2} \int_{\bR^3\setminus \bigcup_k B_{r_k}(r/4)} |\nabla \otimes A|^2\Bigg]\non\\
& \ge \int_\cQ \frac{\rd u}{\ell_u^3}\Bigg\{
 2(2\pi h)^{-3}\iint \big[ (\psi_u \Phi_- \phi_+)(q)\big]^2
\big[ p^2- V_u^+(q)\big]_-\rd q \rd p - C h^{-2+\e}f_u^{4-\e} \ell_u^{2-\e}\Bigg\}\non\\
& = \int_\cQ \frac{\rd u}{\ell_u^3}\Bigg\{
- 2(2\pi h)^{-3}\frac{8\pi}{15}\int \big[ (\psi_u \Phi_- \phi_+)(q)\big]^2
\big[V_u^+(q)\big]_+^{5/2}\rd q \Bigg\} - C h^{-2+\e} \int_\cQ \frac{\rd u}{\ell_u^3}
f_u^{4-\e} \ell_u^{2-\e}.\non
\end{align}
The second term is $O(h^{-2+\e})$ since the integral is finite even 
after extending to $\bR^3$ from $\cQ$. In the first term we use
that the localization errors in $V_u^+$ are bounded by $Ch^2\ell_u^{-2}\le Ch^2[d(x)+r]^{-2}$
and are supported in a ball of radius $CR$,  and
thus
\begin{align}
   h^{-3}\int (\psi_u \Phi_- \phi_+)^2 [ V_u^+]_+^{5/2}
 \le & (1+\e)h^{-3}\int (\psi_u \Phi_- \phi_+)^2 [V]_+^{5/2} + C\e^{-3/2} h^{-3}
\int_{|x|\le CR}\Big[ \frac{h^2}{d(x)+r}\Big]^{5/2}\rd x\non\\
\le & h^{-3}\int  (\psi_u \Phi_- \phi_+)^2 [V]_+^{5/2}+ C\e h^{-3} + C\e^{-3/2}h^2R^{1/2}
\end{align}
since $|V|\le f^2\in L^{5/2}$. Choosing $\e = h^{1+\zeta}$ with a small $\zeta>0$,
and recalling that $R\le Ch^{-1/2}$, we obtain that the two error terms 
are of order $h^{-2+\zeta}$, which even after the $\int_\cQ \ell_u^{-3} \rd u$ integration
is $o(h^{-2})$. 

In the main term, we perform the $\rd u$ integration
and use \eqref{partun} to obtain
$$
  \mbox{Third line of \eqref{3term}}\ge  2(2\pi h)^{-3}\iint \big[ (\Phi_- \phi_+)(q)\big]^2
\big[ p^2- V(q)\big]_-\rd q \rd p - o(h^{-2}).
$$

\medskip

Collecting the estimates of all three terms in  \eqref{3term} and using the properties
of the cutoff functions, we have
\begin{align}
 \tr [T_h(A) & - V]_- +   \frac{1}{\kappa  h^2}
\int_{\bR^3} |\nabla \otimes A|^2  \non\\ 
\ge & \; 
2(2\pi h)^{-3}\iint  \big[ p^2- V(q)\big]_-\rd q \rd p +2h^{-2}\sum_{k=1}^M z_k^2 S(z_k\kappa)
\non\\
& + 2(2\pi h)^{-3}\iint  \Phi_+(q)^2\big[ p^2- V(q)\big]_-\rd q \rd p\non\\
& +  2(2\pi h)^{-3}\sum_{k=1}^M \iint  \theta_{r,k}^2\Bigg( \Big[ p^2- \frac{z_k}{|q-r_k|}\Big]_-
-\Big[ p^2- V(q)\Big]_- \Bigg)\rd q \rd p
 - o(h^{-2}).\label{3term1}
\end{align}
The middle term in the r.h.s. in absolute value is bounded by
$$
 Ch^{-3}\int_{|x|\ge R/2} [V]_+^{5/2} \le \begin{cases}
Ch^{-3}\int_{|x|\ge R/2} d(x)^{-10}\le Ch^{-3} R^7
 \le Ch^{1/2} & \mbox{if $\mu=0$} \\
0 & \mbox{if $\mu\ne 0$.}
\end{cases}
$$
The last term in \eqref{3term1}, also in absolute value, is bounded by
$$
   Ch^{-3}\sum_{k=1}^M \int  \theta_{r,k}^2\Bigg| \Big[\frac{z_k}{|q-r_k|}\Big]^{5/2}
- [V(q)]_+^{5/2} \Bigg|\rd q \le Ch^{-3} \int_{|q|\le r} |q|^{-3/2}\rd q = Ch^{-3}r^{3/2} 
$$
by using \eqref{West}, where $C$ depends on $M$ and $r_{min}$. With our choice of 
$r=h^{1-\frac{1}{10}}$,
this error term is also negligible. This completes the proof
of the lower bound in Theorem~\ref{thm:scMainscott}.

\medskip

\subsection{Upper bound}

We again set $\ell_u=\frac{1}{100}\sqrt{r^2+ d(u)^2}$ and consider
the appropriate cutoff functions $\psi_u$ from \eqref{partun}.
We construct a trial density matrix of the form
\be
 \gamma =\sum_{k=1}^M \theta_{r,k}\gamma_k\theta_{r,k} + 
\int_{\cQ} \frac{\rd u}{\ell_u^3} \phi_+\psi_u\gamma_u \psi_u\phi_+,
\label{trialg}
\ee
where $\cQ$ was defined in \eqref{def:Q} and
$\gamma_k$ and $\gamma_u$ are density matrices to be determined below.
Since
$$
  \sum_k \theta_{r,k}^2 +\phi_+^2 = \phi_-^2+\phi_+^2 =1,
$$
we obtain from \eqref{partun} that $\gamma$ is a density matrix.

\subsubsection{Trial density near the nuclei}\label{sec:trialnear}

To construct $\gamma_k$, we fix some $\eta>0$ and
we notice that from the last part of Theorem~\ref{thm:scott},
for any unscaled cutoff function $\phi$,
there exists some $R(\eta)$ such that for any $R_0>R(\eta)$
there is a vector potential $A_0$, supported
in $B(R_0/4)$, and there is a density matrix $\wh\gamma$ such that
\begin{multline}
\label{def:Sb}
     \tr \Big[ \wh\gamma\phi_{R_0} \Big( T_{h=1}(A_0) - 
 \frac{1}{|x|} \Big) \phi_{R_0}\Big] + \frac{1}{z_k\kappa} \int |\nabla\otimes A_0|^2 
 \\
 - 2(2\pi)^{-3}\int_{\bR^3\times \bR^3} \phi_{R_0}^2(q)
 \Big[ p^2 - \frac{1}{|q|}\Big]_- \rd p \rd q \le 2S(z_k\kappa)+\eta.
\end{multline}
We can also assume that
\be
\label{def:Sbnull}
     \tr \Big[ \wh\gamma\phi_{R_0} \Big( T_{h=1}(A_0) - 
 \frac{1}{|x|} \Big) \phi_{R_0}\Big] + \frac{1}{z_k\kappa} \int |\nabla\otimes A_0|^2 \le0
 \ee
by noticing that the semiclassical integral is of order $R_0^{1/2}$ which
dominates over the $2S(z_k\kappa)+\eta$ term for sufficiently large $R_0$.

Fixing an appropriately large $R_0$,
$\gamma_k$ and $A_k$ are now obtained from $\wh\gamma$ 
and $A_0$ by shifting and rescaling, such that
\begin{multline}
\label{def:Sbk}
     \tr \Big[ \gamma_k\theta_{r,k} \Big( T_{h}(A_k) - 
 \frac{z_k}{|x-r_k|} \Big) \theta_{r,k}\Big]
 + \frac{1}{\kappa h^2} \int |\nabla\otimes A_k|^2 
 \\
 - 2(2\pi h)^{-3}\int_{\bR^3\times \bR^3} \theta_{r,k}^2(q)
 \Big[ p^2 - \frac{z_k}{|q-r_k|}\Big]_- \rd p \rd q \le 2h^{-2}z_k^2S(z_k\kappa)+C h^{-2} \eta
\end{multline}
with $r=R_0h^2z_k^{-1}$.
Here we used that the unscaled cutoff function $\phi$ can be chosen
to be $\phi(x)=\theta_-(d(x))$ so that after rescaling and shift
$\phi_{R_0}$ became $\theta_{r,k}$. We choose $R_0=h^{-1-\xi}z_k$ such that $r=h^{1-\xi}$
and clearly $R_0> R(\eta)$ is satisfied  in the limit as $h\to0$.
We also remark that $A_k$ is supported in $B_{r/4}(r_k)$
which are disjoint balls for different $k$'s. Defining
$A:= \sum_{k=1}^M A_k$, we have $A=A_k$ in the support of $\theta_{r,k}$.
Thus summing up \eqref{def:Sbk}, and using that the replacement
of $z_k|q-r_k|^{-1}$ with $V(q)$ in the semiclassical integral term 
is negligible (see the estimate of the last term in
\eqref{3term1}),
 we have
\begin{multline}
\label{def:Sbk1}
     \sum_{k=1}^M\tr \Big[ \gamma_k\theta_{r,k} \Big( T_{h}(A) - 
 \frac{z_k}{|x-r_k|} \Big) \theta_{r,k}\Big]
 + \frac{1}{\kappa h^2} \int |\nabla\otimes A|^2 
 \\
 \le  2(2\pi h)^{-3}\int_{\bR^3\times \bR^3} \phi_-^2(q)
 \Big[ p^2 - V(q)\Big]_- \rd p \rd q + 2h^{-2}\sum_{k=1}^M
z_k^2S(z_k\kappa)+C h^{-2} \eta.
\end{multline}

Now we establish some properties of the density $\varrho_k(x):=\gamma_k(x,x)$.
Similarly to the estimate $\cT_k^{(3)}(A)$ from \eqref{t3}, we have for any $L>0$
$$
  \inf_A\Bigg\{ \tr \Big[ \phi_{L}\Big( T_{h}(A)- \frac{2}{|x-y|}\Big)\phi_{L}\Big]_- 
 + \frac{1}{\kappa h^2}\int |\nabla \otimes A|^2 \Bigg\}\ge -  Ch^{-3} L^{1/2}
$$
uniformly for any $y\in \bR^3$  if $\kappa$ is sufficiently small.
In particular, for any density matrix $\gamma$ with density $\varrho_\gamma$ we have
$$
 2\sup_{y\in \bR^3} \int \phi_L^2(x)\frac{\varrho_\gamma(x)}{|x-y|}\rd x \le  Ch^{-3} L^{1/2}
+\tr\gamma \phi_{L}T_{h}(A)\phi_{L} 
 + \frac{1}{\kappa h^2}\int |\nabla \otimes A|^2 
$$
for any vector potential $A$ and thus
$$
 \sup_{y\in \bR^3} \int \phi_L^2(x)\frac{\varrho_\gamma(x)}{|x-y|}\rd x \le  Ch^{-3} L^{1/2}
+\inf_A \Bigg\{\tr \Bigg[ \gamma \phi_{L}\Big(T_{h}(A)- \frac{1}{|x|}\Big)\phi_{L}\Bigg] 
 + \frac{1}{\kappa h^2}\int |\nabla \otimes A|^2 \Bigg\}.
$$
Applying this bound to $\wh\gamma$ constructed above
(with $L=R_0$, $h=1$ and $A=A_0$) and using \eqref{def:Sbnull}, we have
$$
 \sup_{y\in \bR^3} \int \phi_{R_0}^2(x)
\frac{\varrho_{\wh \gamma}(x)}{|x-y|}\rd x \le  C R_0^{1/2}
$$
which, after rescaling and shifting $\wh\gamma$ to $\gamma_k$ amounts to
\be
 \sup_{y\in \bR^3} \int \theta_{r,k}^2(x)
\frac{\varrho_{\gamma_k}(x)}{|x-y|}\rd x \le  Ch^{-3}r^{1/2}.
\label{dens1}
\ee
In particular,
\be\label{dens2}
    \int \theta_{r,k}^2(x) \varrho_{\gamma_k}(x) \rd x \le  Ch^{-3}r^{3/2}
\ee
by choosing $y=r_k$ and using that $|x-r_k|\le Cr$ on the support of $\theta_{r,k}$.

Combining \eqref{dens2} with \eqref{West}
 we see that $z_k|x-r_k|^{-1}$ can be replaced with $V$ 
in the l.h.s. of \eqref{def:Sbk1} at an error of order $h^{-3} r^{3/2} = o(h^{-2})$
and thus we have
\begin{multline}
\label{def:Sbk2}
     \sum_{k=1}^M\tr \Big[ \gamma_k\theta_{r,k} \Big( T_{h}(A) - V\Big) \theta_{r,k}\Big]
 + \frac{1}{\kappa h^2} \int |\nabla\otimes A|^2 
 \\
 \le  2(2\pi h)^{-3}\int_{\bR^3\times \bR^3} \phi_-^2(q)
 \Big[ p^2 - V(q)\Big]_- \rd p \rd q + 2h^{-2}\sum_{k=1}^M
z_k^2S(z_k\kappa)+C h^{-2} \eta.
\end{multline}
Moreover, it follows from \eqref{dens1} and \eqref{dens2} that
the density $\theta_{r,k}^2\varrho_{\gamma_k}$ of the density matrix
 $\theta_{r,k}\gamma_k \theta_{r,k}$ 
satisfies
\be
    D( \theta_{r,k}^2\varrho_{\gamma_k}) = \frac{1}{2}\iint  \theta_{r,k}(x)^2 \theta_{r,k}(y)^2
\frac{\varrho_{\gamma_k}(x)\varrho_{\gamma_k}(y)} {|x-y|}\rd x\rd y \le Ch^{-6}r^2.
\label{Dnucl}
\ee

\subsubsection{Trial density away from the nuclei}

Now we construct $\gamma_u$ for any $u\in \cQ$. Since $\mbox{supp}\; A_k \subset B_{r/4}(r_k)$
and $\phi_+(x)$ is supported at $d(x)\ge r/2$, thus on the support of $\phi_+$ we
have $A =0$. Therefore it is sufficient to construct a non-magnetic trial state $\gamma_u$
within each ball $B_u(\ell_u)$ which supports $\psi_u$. This was achieved in
Corollary 15 of \cite{SS} and we just quote the relevant upper bound
(we formulate it particles with spin, this accounts for an additional factor 2
compared with \cite{SS}):
\begin{proposition}\cite[Corollary 15]{SS}\label{prop:SS} Let $\chi \in C_0^7(\bR^3)$ be supported
in $B_\ell$ with some $\ell>0$ and let $V\in C^3(\ov{B_\ell})$ be a real potential. Assume
that for any multiindex $n\in \bN^3$ with $|n|\le 7$ we have
\be
   \|  \partial^n \chi\|_\infty \le C_n \ell^{-|n|}, \qquad
   \|  \partial^n V\|_\infty \le C_n f^{2}  \ell^{-|n|}
\label{psiVb}
\ee
with some constant $f$.
Then there exists a density matrix $\gamma$ such that
$$
  \tr \big[\gamma\chi(-h^2\Delta-V)\chi\big]\le 2(2\pi h)^{-3} \iint \chi(q)^2[p^2- V(q)]_-
\rd p\rd q + C h^{-3+6/5}f^{3+4/5} \ell^{3-6/5},
$$
where $C$ depends only on the constants in \eqref{psiVb}. Moreover, the density
$\varrho_\gamma(x)$ of $\gamma$ satisfies
\be
  \Big| \varrho_\gamma(x) -2(2\pi h)^{-3}\om_3 [V(x)]_+^{3/2}\Big|\le C h^{-3+9/10} f^{3-9/10}
  \ell^{-9/10},
\label{densbb}
\ee 
for almost all $x\in B_\ell$, and
\be
  \Big| \int \chi^2\varrho_\gamma - 2(2\pi h)^{-3} \om_3
  \int \chi^2[V]_+^{3/2} \Big|\le Ch^{-3+6/5},
\label{densintb}
\ee
where $\om_3= 4\pi/3$ is the volume of the unit ball.
\end{proposition}

We will apply  this Proposition for each ball $B_u(\ell_u)$
and with $\chi:= \phi_+\psi_u$.
{F}rom \eqref{psider} and \eqref{Vder} it
 is easy to see that the conditions \eqref{psiVb} are satisfied
with our choice of $\ell =\ell_u =  \frac{1}{100}\sqrt{r^2+ d(u)^2}$
and $f_u=f(u)=\min\{ d(u)^{-1/2}, d(u)^{-2}\}$ as given in
\eqref{fdef}. The density matrix  thus constructed in Proposition~\ref{prop:SS}
will be denoted by $\gamma_u$ and this will be the density matrix
 in \eqref{trialg}. With this choice we have
\begin{align}\label{fullinn}
\tr \Big[ & (-h^2\Delta -V)  
\int_{\cQ} \frac{\rd u}{\ell_u^3} \phi_+\psi_u\gamma_u \psi_u\phi_+\Big]\non\\
 \le &\; 2(2\pi h)^{-3} \int_{\cQ} \frac{\rd u}{\ell_u^3} \iint \phi_+(q)^2\psi_u(q)^2
 [p^2- V(q)]_- \rd p\rd q 
 + C h^{-3+6/5}\int_{\cQ} \frac{\rd u}{\ell_u^3}  f_u^{3+4/5} \ell_u^{3-6/5}\non\\
 \le & \; 2(2\pi h)^{-3} \iint \phi_+(q)^2
 [p^2- V(q)]_- \rd p\rd q 
 + C h^{-2+1/5}\int_{\cQ} \frac{\rd u}{\ell_u^3}  f_u^{3+4/5} \ell_u^{3-6/5}.
\end{align}
The error term can easily be computed by using that $\ell_u \sim d(u)$ for $u\in \cQ$
as
\begin{align}
  \int_{\cQ} \frac{\rd u}{\ell_u^3}  f_u^{3+4/5} \ell_u^{3-6/5}
  \le & \; C\int_{r/3\le d(u)\le 1} d(u)^{-19/10 - 6/5}  \rd u
  + C\int_{1\le d(u)\le 2R} d(u)^{-38/5 - 6/5}  \rd u \non\\
\le &\;  Cr^{-1/10} +C.
\end{align}
Since $r\ge h$, we have $h^{1/5} r^{-1/10}\le h^{1/10}$, so the error term 
in \eqref{fullinn} is $o(h^{-2})$. 

We now use $\gamma$ from \eqref{trialg} and $A=\sum_k A_k$ constructed in
Section~\ref{sec:trialnear} to complete the proof of
the upper bound \eqref{trialenergy} in Theorem~\ref{thm:scMainscott}.
Combining \eqref{def:Sbk2} and  \eqref{fullinn} and recalling 
that $A$ and $\phi_+$ are supported disjointly, we have
\begin{align}
\tr[T_h(A) & -V]\gamma + \frac{1}{\kappa h^2} \int |\nabla\otimes A|^2\non\\
&\le  2(2\pi h)^{-3}\iint \big[ p^2 - V(q)\big]_- \rd q \rd p
 +  2h^{-2}\sum_{k=1}^M z_k^2 S(z_k\kappa) +o(h^{-2})+Ch^{-2}\eta, \non
\end{align}
where we also used that $\phi_+^2+\phi_-^2=1$.
Finally, letting $h\to 0$ first and then $\eta\to0$, we obtain
the upper bound in \eqref{locscMain1} and  \eqref{trialenergy}.

\medskip

To complete the proof of Theorem~\ref{thm:scMainscott}, it remains to prove
\eqref{densitycontroll} and \eqref{densitycontroll65}. 
{F}rom \eqref{trialg} we have 
$$
 \varrho_\gamma = \sum_{k=1}^M \theta_{r,k}^2 \varrho_{\gamma_k} + 
  \int_{\cQ} \frac{\rd u}{\ell_u^3} \psi_u^2 \phi_+^2 \varrho_{\gamma_u}.
$$
Using  \eqref{dens2} and $\phi_+\le 1$ we have
$$
  \int\varrho_\gamma\le Ch^{-3} r^{3/2} +  
\int_{\cQ} \frac{\rd u}{\ell_u^3} \int \psi_u^2 \varrho_{\gamma_u}.
$$
The first term is of order $h^{-3/2(1+\xi)}= h^{-2+7/20}$, hence negligible. In the second term 
we use \eqref{densintb} and \eqref{partun} to get
\begin{align}
   \int_{\cQ} \frac{\rd u}{\ell_u^3} \int \psi_u^2 \varrho_{\gamma_u}
   & \le \frac{1}{3\pi^2h^3}\int_{\cQ} \frac{\rd u}{\ell_u^3} \int \psi_u^2 [V]_+^{3/2}
  + Ch^{-2+1/5} \int_{\cQ} \frac{\rd u}{\ell_u^3}\non\\
&\le  \frac{1}{3\pi^2h^3}\int [V]_+^{3/2}+ Ch^{-2+1/5}(|\log r|+ |\log R|).
\end{align}
Since both logarithms are of order $|\log h|$, we obtained \eqref{densitycontroll}.

To prove \eqref{densitycontroll65}, we note that $\sqrt{ D(\varrho)}$ satisfies the
triangle inequality, thus we have
\begin{align}
 D\Big(\varrho_\gamma - (3\pi^2)^{-1} h^{-3} [V]_-^{3/2}\Big) 
   \le & \; C D\Big( \int_{\cQ} \frac{\rd u}{\ell_u^3} \psi_u^2 \phi_+^2 \Big[ \varrho_{\gamma_u}
  - \frac{1}{3\pi^2 h^3} [V]_+^{3/2}\Big]\Big)\label{3D} \\
& + C  D\Big(  \frac{1}{3\pi^2 h^3} [V]_+^{3/2}
\Big[ \int_{\cQ} \frac{\rd u}{\ell_u^3} \psi_u^2 \phi_+^2 -1
\Big]\Big)
+ C \sum_k D(\varrho_{\gamma_k}).  \non
\end{align}
The last term is smaller than $O(h^{-5+\e})$ with some small $\e$ by 
using \eqref{Dnucl}.
For the second term in \eqref{3D} we note that by \eqref{partun}
$$
  \Big| \int_{\cQ} \frac{\rd u}{\ell_u^3} \psi_u(x)^2 \phi_+(x)^2 -1
\Big| \le |\phi_+(x)^2-1| + \int_{\cQ^c} \frac{\rd u}{\ell_u^3} \psi_u(x)^2
\le C\big( {\bf 1}( d(x)\le r) + {\bf 1}( d(x)\ge R) \big)
$$
and thus
$$
  D\Big(  \frac{1}{3\pi^2 h^3} [V]_+^{3/2}
\Big[ \int_{\cQ} \frac{\rd u}{\ell_u^3} \psi_u^2 \phi_+^2 -1
\Big]\Big) \le Ch^{-6} \Big[ D\Big( [V]_+^{3/2} {\bf 1}( d(x)\le r)\Big)
  +   D\Big( [V]_+^{3/2} {\bf 1}( d(x)\ge R)\Big)\Big].
$$
By the Hardy-Littlewood-Sobolev inequality we have $D(\varrho)\le C\|\varrho\|_{6/5}^2$
for any real function $\varrho$, therefore this error is bounded
by
$$
 Ch^{-6}\Big( \int_{d(x)\le r} [V]_+^{9/5} +  \int_{d(x)\ge R} [V]_+^{9/5}\Big)^{5/3}.
$$
Using that $V(x)$ is essentially $z_k|x-r_k|^{-1}$ near the nuclei and
$V(x) \sim |x|^{-4}$ for large $x$ (see \eqref{Vder} and \eqref{West}), we
easily obtain 
$$
  D\Big(  \frac{1}{3\pi^2 h^3} [V]_+^{3/2}
\Big[ \int_{\cQ} \frac{\rd u}{\ell_u^3} \psi_u^2 \phi_+^2 -1
\Big]\Big) \le Ch^{-6}\big(r^2 + R^{-7}\big)  
$$
which is smaller than   $O(h^{-5+\e})$.

Finally, we estimate the first term on the r.h.s. of \eqref{3D}.
By  the Hardy-Littlewood-Sobolev inequality and \eqref{densbb}
\begin{align}
 D\Big( \int_{\cQ} \frac{\rd u}{\ell_u^3} \psi_u^2 \phi_+^2 \Big[ \varrho_{\gamma_u}
  - \frac{1}{3\pi^2 h^3} [V]_+^{3/2}\Big]\Big)^{1/2}
\le & C \, \Bigg\| \int_{\cQ} \frac{\rd u}{\ell_u^3} \psi_u^2 \phi_+^2 \Big[ \varrho_{\gamma_u}
  - \frac{1}{3\pi^2 h^3} [V]_+^{3/2}\Big] \Bigg\|_{6/5}\non\\
\le &  \int_{\cQ} \frac{\rd u}{\ell_u^3} \Bigg\| \psi_u^2 \phi_+^2 \Big[ \varrho_{\gamma_u}
  - \frac{1}{3\pi^2 h^3} [V]_+^{3/2}\Big] \Bigg\|_{6/5} \non\\
\le & Ch^{-2-\frac{1}{10}} \int_{\cQ} \frac{\rd u}{\ell_u^3} f_u^{21/10} \ell_u^{8/5}
 \non\\
\le & Ch^{-2-\frac{1}{10}}
\end{align}
as the last integral is bounded.
In the second line we estimated
$$
   \Bigg\| \psi_u^2 \phi_+^2 \Big[ \varrho_{\gamma_u}
  - \frac{1}{3\pi^2 h^3} [V]_+^{3/2}\Big] \Bigg\|_{6/5} 
  \le Ch^{-2-\frac{1}{10}} f_u^{21/10} \ell_u^{8/5}
$$
by using  \eqref{densbb} and \eqref{psider}. This completes the
proof of  \eqref{densitycontroll65} and the proof of  Theorem~\ref{thm:scMainscott}.
\qed

\section{Equivalence of the two definitions of $S(\kappa)$}\label{sec:equiv}

{\it Proof of Lemma~\ref{S=S}.} 
Setting $h=\nu^{1/2}$, 
after a change of variables we have
$$
\frac{2}{(2\pi)^3} \int_{\bR^3\times \bR^3}
 \Big[ p^2 - \frac{1}{|q|}+\nu\Big]_- \rd p \rd q =
\nu\frac{2}{(2\pi h)^3} \int_{\bR^3\times \bR^3}
 \Big[ p^2 - \frac{1}{|q|}+1\Big]_- \rd p \rd q .
$$
Similarly,  by  a simple rescaling, $x\to \nu^{-1}x$ we get
$$
 \inf_A \Big\{ \tr \Big[  T_{1}(A) - 
 \frac{1}{|x|} +\nu\Big]_- + \frac{1}{\kappa} \int_{\bR^3} |\nabla\otimes A|^2 
\Big\} = \nu \inf_A \Big\{ \tr \Big[  T_{\nu^{1/2}}(A) - 
 \frac{1}{|x|} +1\Big]_- + \frac{1}{\kappa \nu} \int_{\bR^3} |\nabla\otimes A|^2 
\Big\}.
$$
We apply Theorem~\ref{thm:scMainscott} to the potential 
$V(x)= |x|^{-1} -1$ and with $h=\nu^{1/2}$. Notice that $V(x)$
satisfies the conditions \eqref{Vder} and \eqref{West} with $\mu=1$, $M=1$,
$r_{min}=\infty$, $z_1=1$ and $r_1=0$. We get
$$
\inf_A \Big\{ \tr \Big[  T_{1}(A) - 
 \frac{1}{|x|} +\nu\Big]_- + \frac{1}{\kappa} \int_{\bR^3} |\nabla\otimes A|^2 
\Big\}  - \frac{2}{(2\pi)^3} \int_{\bR^3\times \bR^3}
 \Big[ p^2 - \frac{1}{|q|}+\nu\Big]_- \rd p \rd q = 2S(\kappa) + O(\nu^{\e/2})
$$
which proves \eqref{def:Skappa1}.
\qed


\begin{thebibliography}{hhhhh}





\bibitem[ES3]{ES3}  L. Erd{\H o}s, J.P. Solovej, 
{\em Ground state energy of large atoms in a self-generated
magnetic field.} Commun. Math. Phys. {\bf 294}, No. 1, 229-249 (2009)




\bibitem[EFS1]{EFS1}  L. Erd{\H o}s, S. Fournais, J.P. Solovej: 
{\em Stability and semiclassics in self-generated
fields.} Preprint, available at arxiv.org.

\bibitem[EFS2]{EFS2}  L. Erd{\H o}s, S. Fournais,
J.P. Solovej: {\em Second order semiclassics with self-generated
magnetic fields.} Preprint,  available at arxiv.org.




\bibitem[FS]{FS} C.~Fefferman and L.A.~Seco: {\em On the energy of a
    large atom}, Bull.~AMS {\bf 23}, 2, 525--530 (1990).

\bibitem[FSW1]{FSW1} R. L. Frank, H. Siedentop, S. Warzel:
{\em The ground state energy of heavy atoms: 
relativistic lowering of the leading energy correction.} 
Commun. Math. Phys.  {\bf 278}
   no. 2, 549-–566 (2008)




\bibitem[FSW2]{FSW2} R. L. Frank, H. Siedentop, S. Warzel:
{\em The energy of heavy atoms according to Brown and Ravenhall: the Scott correction.}
Doc. Math. {\bf 14}, 463--516 (2009).







\bibitem[FLL]{FLL}
J. Fr\"ohlich, E. H. Lieb, and M. Loss:
{\it Stability of Coulomb systems with magnetic fields. 
I. The one-electron atom. }
Commun.\ Math.\ Phys.\ {\bf 104} 251--270 (1986) 

\bibitem[H]{H} W.~Hughes: {\em An atomic energy bound that gives
    Scott's correction}, Adv.~Math. {\bf 79}, 213--270 (1990).



\bibitem[IS]{IS} V.I.~Ivrii and I.M.~Sigal: {\em Asymptotics of
    the ground state energies of large Coulomb systems}, Ann.~of Math.
  (2), {\bf 138}, 243--335 (1993).


\bibitem[L]{L} E. H. Lieb: {\it Thomas-Fermi and related
theories of atoms and molecules}, Rev. Mod. Phys. {\bf 65}. No. 4, 603-641
(1981)



\bibitem[L2]{L2} E. H. Lieb: {\it Variational principle for
many-fermion systems}, Phys. Rev. Lett. {\bf 46}, 457--459 (1981)
and {\bf 47} 69(E) (1981)



\bibitem[LLS]{LLS} E. H. Lieb, M. Loss and J. P. Solovej: {\em
Stability of Matter in Magnetic Fields}, Phys. Rev. Lett. {\bf 75},
 985--989 (1995)


\bibitem[LO]{LO} E. H. Lieb and S. Oxford: {\it Improved Lower Bound
    on the Indirect Coulomb Energy}, Int. J. Quant. Chem. {\bf 19},
  427--439, (1981)




\bibitem[LS]{LS} E. H. Lieb and B. Simon: {\em The Thomas-Fermi theory
of atoms, molecules and solids}, Adv. Math. {\bf 23}, 22-116 (1977)




\bibitem[SW1]{SW1} H.~Siedentop and R.~Weikard: {\em On the leading
    energy correction for the statistical model of an atom:
    interacting case}, Commun.~Math.~Phys.~ {\bf 112}, 471--490
  (1987) 

\bibitem[SW2]{SW2} H.~Siedentop and R.~Weikard: {\em On the leading
    correction of the Thomas-Fermi model: lower bound}, Invent.~Math.
  {\bf 97}, 159--193 (1990)

\bibitem[SW3]{SW3} H.~Siedentop and R.~Weikard:
 {\em A new phase space localization technique with
    application to the sum of negative eigenvalues of {S}chr\"odinger
    operators}, Ann.~Sci.~\'Ecole Norm. Sup. (4), {\bf 24}, no.~2,
  215--225 (1991).
  
 \bibitem[Sob1]{Sob1} A.~V.~Sobolev: {\em Quasi-classical asymptotics
of local Riesz means for the Schr\"odinger operator in a moderate
magnetic field.} Ann. Inst. H. Poincar\'e, {\bf 62}  no. 4, 325-360,  (1995)

 \bibitem[Sob]{Sob} A.~V.~Sobolev: {\em Discrete spectrum 
asymptotics for the Schr\"{o}dinger operator with a
 singular potential and a magnetic field},
 Rev.~Math.~Phys {\bf 8} (1996) no.~6, 861--903.
  

\bibitem[SS]{SS}  J. P. Solovej, W. Spitzer:  {\em A new
coherent states approach to semiclassics which gives
Scott's correction.} Comm. Math. Phys.  {\bf 241}  (2003),  no. 2-3, 383--420.

\bibitem[SSS]{SSS}  J. P. Solovej, T.\O.  S\o rensen, W. Spitzer:  {\it 
Relativistic Scott correction for atoms and molecules.}
Comm. Pure Appl. Math. Vol. LXIII. 39-118 (2010).






\end{thebibliography}
\end{document}